\documentclass[journal]{IEEEtran}

 \newcommand{\MATLAB}{\textsc{Matlab}\xspace}
\usepackage[T1]{fontenc} 
\ifCLASSINFOpdf
\else
\fi
\usepackage{bm,color,soul}
\usepackage{amsmath, amsthm, amssymb}
\usepackage{caption}
\usepackage{subcaption}
 
\theoremstyle{plain}
\newtheorem{thm}{Theorem}[]

\usepackage[pdftex]{graphics}

\theoremstyle{definition}

\usepackage{xspace}
\usepackage{amssymb}
\usepackage{graphicx}
\usepackage{amsmath}
\usepackage{mathtools}
\usepackage{graphicx,color} 
\usepackage{amsmath, amsfonts, amssymb, amsbsy,nccmath} 
\usepackage{algorithm} 
\usepackage{enumerate} 
\usepackage{algorithmic} 
\usepackage{lipsum} 
\usepackage[sort,compress]{cite} 
\usepackage{epsfig} 
\usepackage{epstopdf}
\usepackage{mathtools}
\usepackage{dsfont} 
\usepackage[inline]{enumitem}
\usepackage{amsmath}
\usepackage[cmintegrals]{newtxmath}
\usepackage{color,soul}
\usepackage{stfloats}  
\usepackage{tabularx}
\usepackage{lipsum}
\usepackage{mwe}
\usepackage{float}

\newcommand{\bmA}{\mathbf A}

\newcommand{\bmI}{\mathbf I}
\newcommand{\bmT}{\mathbf T}
\newcommand{\bmQ}{\mathbf Q}

\newcommand{\bmD}{\mathbf D}
\newcommand{\bmL}{\mathbf L}
\newcommand{\bmP}{\mathbf P}
\newcommand{\bmG}{\mathbf G}
\newcommand{\bmH}{\mathbf H}
\newcommand{\bmF}{\mathbf F}

\newcommand{\bms}{\mathbf s}
\newcommand{\bmV}{\mathbf V}
\newcommand{\bmU}{\mathbf U}
\newcommand{\bmy}{\mathbf y}
\newcommand{\bmR}{\mathbf R}
\newcommand{\bmX}{\mathbf X}
\newcommand{\bmW}{\mathbf W}
\newcommand{\bmx}{\mathbf x}
\newcommand{\bmn}{\mathbf n}
\newcommand{\bmc}{\mathbf c}
\newcommand{\bme}{\mathbf e}
\newcommand{\bmr}{\mathbf r}

\begin{document}

\title{Parallel and Distributed Hybrid Beamforming for Multicell Millimeter Wave MIMO Full Duplex}  
\author{Chandan~Kumar~Sheemar, Symeon Chatzinotas, Dirk~Slock, Eva Lagunas and Jorge Querol
% <-this % stops a space
 \thanks{ Chandan Kumar Sheemar, Symeon Chatzinotas, Eva Lagunas and Jorge Querol are with the SnT department at the University of Luxembourg (email:\{chandankumar.sheemar,symeon.chatzinotas,eva.lagunas,jorge.querol\}\\@uni.lu). Dirk Slock is with the communication systems department at EURECOM, Sophia Antipolis, France (email:slock@eurecom.fr). } 
 %\thanks{A preliminary version of this work is part of the PhD thesis \cite{sheemar2022hybrid_thesis} and is available on arXiv at }\vspace{-4mm}} 
 }
 \maketitle

 %This article presents two novel hybrid beamforming (HYBF) designs for a multicell massive multiple-input-multiple-output (mMIMO) millimeter wave (mmWave) full duplex (FD) system under limited dynamic range (LDR). Firstly, a novel centralized HYBF (C-HYBF) scheme based on the minorization-maximization (MM) method is presented. However, C-HYBF presents many drawbacks such as high computational complexity, massive communication overhead to transfer complete channel state information (CSI) to the central node (CN) every channel coherence time (CCT), and requirement of expensive computational resources for joint optimization. To overcome these drawbacks, we present a very low-complexity, per-link parallel and distributed HYBF (P$\&$D-HYBF) scheme based on cooperation. Due to the per-link independent decomposition, it enables each FD base station (BS) to solve its local sub-problems independently and in parallel on multiple processors, which leads to significant reduction in the communication overhead. It requires that each FD BS cooperates by exchanging information about its beamformers, which allows each FD BS to adapt its beamformers correctly, leading to a negligible performance loss compared to C-HYBF. Simulation results show that both designs achieve similar performance and outperform the fully digital half duplex (HD) system with only a few radio-frequency (RF) chains. 
\begin{abstract}
Full duplex (FD) is an auspicious wireless technology that holds the potential to double data rates through simultaneous transmission and reception. This paper proposes two innovative designs of hybrid beamforming (HYBF) for a multicell massive multiple-input-multiple-output (mMIMO) millimeter wave (mmWave) FD system. Initially, we introduce a novel centralized HYBF (C-HYBF) scheme, which employs the minorization-maximization (MM) method. However, while centralized beamforming designs offer superior performance, they suffer from high computational complexity, substantial communication overhead, and demand expensive computational resources. To surmount these challenges, we present a framework that facilitates per-link parallel and distributed HYBF (P$\&$D-HYBF) in the mmWave frequency band. This cooperative approach enables each base station (BS) to independently solve its local, low-complexity sub-problems in parallel, resulting in a substantial reduction in communication overhead and computational complexity. Simulation results demonstrate that P$\&$D-HYBF achieves comparable performance to C-HYBF, and with only a few radio-frequency (RF) chains, both designs surpass the capabilities of fully digital half duplex (HD) systems.
\end{abstract}
\begin{IEEEkeywords}
Millimeter Wave, Full duplex,  Distributed Hybrid Beamforming, Minorization Maximization
\end{IEEEkeywords}

 \vspace{-3mm}

\IEEEpeerreviewmaketitle

\section{Introduction} \label{Intro}
\IEEEPARstart{C}{ellular} communication networks are in a perpetual state of evolution, driven by the escalating demands of emerging wireless data services and the need to maintain uninterrupted connectivity. To accommodate the anticipated surge in data volume, the research community has recently shown significant interest in employing ultra-high-frequency bands, notably the millimeter wave (mmWave) band \cite{wang2018millimeter}. In comparison to traditional radio frequency (RF) and microwave bands, the mmWave band provides approximately 200 times greater spectrum. However, mmWave communication systems encounter substantial propagation losses, which are mitigated by employing a substantial number of antennas that can be densely packed due to the shorter wavelength.

In parallel to the development of mmWave half duplex (HD) systems has been the progression of full duplex (FD) technology which enables simultaneous transmission and reception in the same frequency band, which theoretically doubles the spectral efficiency. 
However, FD systems suffer from self-interference (SI), which can be $90-120$ dB higher than the received signal power. Advanced SI cancellation (SIC) techniques are crucial to mitigate the SI power and make FD feasible. Among the techniques being explored for SI management in FD systems, beamforming has exhibited significant potential. Beamforming can be categorized as either fully digital or hybrid beamforming (HYBF). The former necessitates an equal number of radio frequency (RF) chains and antennas, processing the signal in the baseband. On the other hand, the latter requires only a limited number of RF chains and employs lower-dimensional digital processing alongside higher-dimensional analog processing. The implementation of such a solution is highly desirable, as it enables the deployment of cost-effective mmWave FD transceivers \cite{roberts2021millimeter_survey}.

%Hybrid beamforming (HYBF) can enable the mmWave FD transceiver design with a less number of RF chains 
%and is a promising tool also for handling the SI \cite{roberts2021millimeter_survey}.
%however, in practice, its maximum achievable gain remains limited by the non-ideal hardware, including power
%amplifiers (PAs), analog-to-digital-converters (ADCs), etc. Such non-idealities can be considered in the beamforming designs with the limited dynamic range (LDR) model \cite{cirik2015weighted}.  

%Such effect can be considered in the beamforming designs with the LDR noise model \cite{day2012full}, which allows to investigate the effective performance gain of the practical FD systems.

\subsection{Prior Work and Motivation}
Recent studies investigating the potential of HYBF for mmWave FD systems are available in \cite{palacios2019hybrid,lopez2019analog,sheemar2021hybrid_interference,cai2020two,luo2021robust,sheemar2021massive,lopez2019beamformer,da20201,roberts2020hybrid,sheemar2021practical_HYBF}.
In \cite{palacios2019hybrid}, HYBF for a point-to-point massive multiple-input multiple-output (mMIMO) mmWave FD system to optimize the sum-spectral efficiency while keeping the signal level at the input of the analog-to-digital converters (ADCs) under control is investigated. In \cite{lopez2019analog}, a novel
algorithm for the design of constant-amplitude analog precoders
and combiners of an FD mmWave single-stream bidirectional
link is discussed. In \cite{sheemar2021hybrid_interference}, HYBF for an FD mmWave MIMO interference channel is presented. In \cite{cai2020two}, the authors studied the performance of the two-timescale HYBF approach for the mmWave FD relays. Robust HYBF under imperfect channel state information (CSI) is studied in \cite{luo2021robust}. In \cite{sheemar2021massive}, HYBF for integrated access and backhaul is investigated. In \cite{lopez2019beamformer}, the authors presented a novel HYBF design for an amplify and forward mmWave FD relay. In \cite{da20201}, a novel HYBF design for a single-antenna multi-user mmWave FD system with $1$-bit phase resolution is proposed. In \cite{roberts2020hybrid}, the achievable gain of a mmWave FD system with one uplink (UL) and one downlink (DL) user only under the limited receive dynamic range is investigated. In \cite{sheemar2021practical_HYBF}, the potential of HYBF for a single-cell mMIMO mmWave FD system is studied.  

It is noteworthy that the existing literature has not addressed the context of mmWave multi-cell FD systems with multi-antenna users. Furthermore, all the proposed methods thus far have been centralized, which presents significant challenges in terms of feasibility, scalability, and cost-effectiveness for large-scale real-time deployment. Centralized solutions, in general, also suffer from substantial communication overhead, as the global channel state information (CSI) must be transmitted to the central node (CN) at each channel coherence time. This transmission may involve multi-hop communication if the CN is situated far from the network. Then, the CN has to perform joint optimization of all variables on a high-performance computational processor, followed by transmitting the optimized solution back to the base stations (BSs). Such a procedure, necessitated at the millisecond scale, is clearly impractical.
Distributed solutions hold great promise in addressing these challenges; however, their design is exceptionally complex \cite{scutari2013decomposition,palomar2006tutorial,scutari2016parallel}, particularly in the mmWave domain where the analog beamformers and combiners are shared among the users \cite{abrardo2019distributed}.

\subsection{Main Contributions}

%In what follows, to fill the aforementioned research gaps, we present two HYBF designs for a multicell mMIMO mmWave FD system with multi-antenna UL and DL users for weighted sum rate (WSR) maximization. Firstly, we aim at presenting a novel centralized HYBF (C-HYBF) scheme based on the minorization-maximization (MM) method \cite{stoica2004cyclic}, relying on alternating optimization. Remark that, although suffering from many inherent challenges at the implementation level, centralized solutions exhibit superior 
%performance and can serve as a benchmark to compare the performance of distributed methods \cite{sheemar2021game}. To overcome the challenges of C-HYBF, we for the first time ever, present the concept of parallel and distributed (P$\&$D)-HYBF for mmWave. Such a solution is very promising as it decomposes the global multi-cell FD WSR maximization problem into per-link independent and fully decoupled low-complexity optimization sub-problems. Thus, it can be implemented by each mmWave FD BS on very low-cost computational processors in parallel and independently, leading to an extremely efficient approach for HYBF. P$\&$D-HYBF is based on cooperation, which requires information sharing about the beamformers, but only among the neighbouring BSs. Moreover, in contrast to the C-HYBF, which requires full CSI to be transferred to the CN, P$\&$D-HYBF relies only on the local CSI.
In the subsequent sections, aiming to bridge the identified research gaps, we present two HYBF designs for a multicell mmWave FD system with multiple antenna UL and DL users, focusing on weighted sum rate (WSR) maximization. Initially, we introduce a novel centralized HYBF (C-HYBF) scheme, leveraging the minorization-maximization (MM) method based on alternating optimization \cite{stoica2004cyclic}. It is worth noting that while centralized solutions entail inherent challenges at the implementation level, they demonstrate superior performance and can serve as a benchmark for comparing the performance of distributed methods \cite{sheemar2021game}.

To overcome the challenges posed by C-HYBF, we introduce the concept of parallel and distributed (P$\&$D)-HYBF for mmWave systems, a pioneering approach that decomposes the global mmWave FD WSR maximization problem into independent and decoupled low-complexity optimization sub-problems on a per-link basis. This approach allows each mmWave FD base BS to implement the optimization independently and in parallel, utilizing cost-effective computational processors. As a result, P$\&$D-HYBF offers an exceptionally efficient solution for HYBF. The implementation of P$\&$D-HYBF relies on cooperation, involving information sharing of beamformers solely among neighboring BSs. Furthermore, unlike C-HYBF, which necessitates the transmission of full CSI to the CN, P$\&$D-HYBF relies solely on local CSI, thereby significantly reducing the overhead associated with CSI exchange.

The design of P$\&$D-HYBF in mmWave systems presents considerable challenges due to the shared utilization of analog beamformers and combiners by the DL and UL users within the same cell. This sharing introduces intricately coupled constraints that significantly complicate per-link decoupling and optimization. Furthermore, the DL beamformers are subject to a total sum-power constraint that further exacerbates the complexity of the system. Our objective is to address these challenges and establish a framework for P$\&$D-HYBF in mmWave systems by elucidating strategies to effectively manage the various coupling constraints within the hybrid beamforming architecture. The computational analysis demonstrates that P$\&$D-HYBF exhibits remarkably low complexity, which scales linearly with the size and density of the network. This stands in contrast to the centralized HYBF (C-HYBF) scheme, which exhibits a quadratic dependence on both network size and density, necessitating exceptionally high computational complexity per iteration. Simulation results validate that P$\&$D-HYBF achieves performance on par with C-HYBF while surpassing the conventional fully digital half-duplex (HD) system, even with a limited number of RF chains.

 \emph{Organization:} The rest of the paper is organized as follows: We first present the system model and problem formulation in Section \ref{system_model}, and the MM method is presented in Section \ref{simplificazione_problem}. Sections \ref{Centralized_HYBF} and \ref{Distributed_HYBF} present the C-HYBF scheme and the P$\&$D-HYBF designs, respectively. Finally, Sections \ref{simulazioni} and \ref{conclusioni} discuss the simulation results and conclusions, respectively.
 
 \emph{Mathematical Notations:}  Boldface lower and upper case case characters denote vectors and matrices, respectively. $\mathbb{E}\{\cdot\}, \mbox{Tr}\{\cdot\}, (\cdot)^H$, $\otimes$, $\bmI$, and $\bmD_{d}$
denote expectation, trace, conjugate transpose, Kronecker product, identity matrix, and the $d$ generalized dominant eigenvectors (GDEs) selection matrix, respectively. A vector of zeros of size $M$ is denoted as $\mathbf{0}_{M \times 1}$, $\mbox{vec}(\bmX)$ stacks the column of $\bmX$ into $\bmx$, $\mbox{unvec}(\bmx)$ reshapes $\bmx$ into $\bmX$, and $\angle \bmX$ returns the phasors of matrix $\bmX$. $\mbox{Cov}(\cdot)$ and diag$(\cdot)$ denote the covariance and diagonal matrices, respectively, and $svd(\bmX)$ returns the singular value decomposition (SVD) of $\bmX$. $\bmX(m,n)$ denotes the element at the m-th row and n-th column.

\section{System Model} \label{system_model}

\begin{figure}
     \centering
\includegraphics[width=7cm,height=4.5cm,draft=false]{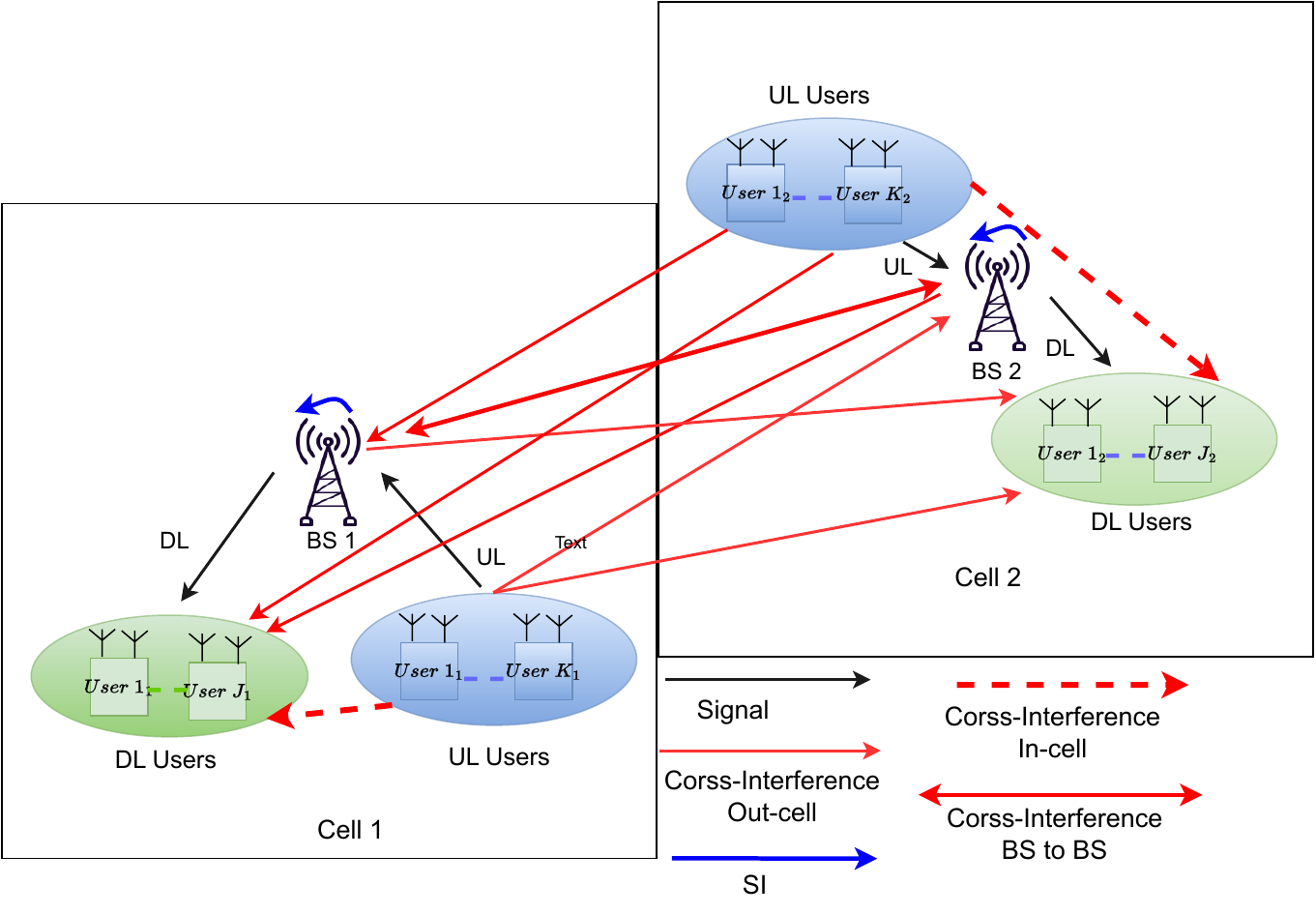}
     \caption{The multicell mmWave MIMO FD system.}
     \label{scenario_FD} \vspace{-5mm}
\end{figure}
Let $\mathcal{B} = \{1,....,B\}$ denote the set containing the indices of $B$ FD BSs serving in $B$ cells. Let $\mathcal{D}_b= \{1,...,D_b\}$ and $\mathcal{U}_b= \{1,...,U_b\}$ denote the sets containing the indices of $D_b$ DL and $U_b$ UL multi-antenna HD users communicating with BS $b \in \mathcal{B}$. The system setting with $B=2$ is shown in Fig. \ref{scenario_FD}. The DL user $j_b \in \mathcal{D}_b$ and UL user $k_b \in \mathcal{U}_b$ are assumed to have $N_{j_b}$ receive and $M_{k_b}$ transmit antennas, respectively. The FD BS $b \in \mathcal{B}$ is assumed to have $M_{b}^{RF}$ and $N_{b}^{RF}$ transmit and receive RF chains and $M_b$ and $N_b$ transmit and receive antennas, respectively. 
We denote with $\bmV_{j_b} \in \mathbb{C}^{M_{b}^{RF} \times d_{j_b}}$ and $\bmU_{k_b} \in \mathbb{C}^{M_{k_b} \times d_{k_b}}$ the digital beamformers for the white unitary variance data streams $\bms_{j_b} \in \mathbb{C}^{d_{j_b} \times 1}$ and $\bms_{k_b} \in \mathbb{C}^{d_{k_b} \times 1}$ transmitted for DL user $j_b \in \mathcal{D}_b$ and from UL user $k_b \in \mathcal{U}_b$, respectively. Let $\bmW_b \in \mathbb{C}^{M_{b} \times M_{b}^{RF}} $ and $\bmF_b \in \mathbb{C}^{N_{b}^{RF} \times N_{b}}$ denote the fully connected analog beamformer and analog combiner for FD BS $b \in \mathcal{B}$, respectively. The analog stage is assumed to be quantized and let $\mathcal{P}_b = \{1, \;e^{i 2 \pi/n_b},\; e^{i 4 \pi/n_b},...,e^{i 2 \pi n_b-1/n_b}\}$  denote the set of $n_b$ possible discrete values that the unit-modulus elements of $\bmW_b$ and $\bmF_b$ can assume. Let $\mathbb{Q}_b(\cdot)$ denote the quantizer function to quantize the infinite resolution elements of $\bmW_b$ and $\bmF_b$ such that $\mathbb{Q}_b(\angle \bmW_b(m,n))$ and $\mathbb{Q}_b(\angle \bmF_b(m,n) \in  \mathcal{P}_b, \forall m,n$.

We assume the users to be suffering from hardware distortions due to limited dynamic range (LDR), which is denoted as $\bmc_{k_b}$ and $\bme_{j_b}$ for the UL user $k_b \in \mathcal{U}_b$ and DL user $j_b \in \mathcal{D}_b$, respectively, and modelled with the LDR noise model \cite{cirik2015weighted}

 \begin{equation}  
    \bmc_{k_b} \sim \mathcal{CN}(\mathbf{0}_{M_{k_b} \times 1}, k_{k_b}\; \mbox{diag}( \bmU_{k_b} \bmU_{k_b}^H  )), 
\end{equation}
 \begin{equation}  
        \bme_{j_b} \sim \mathcal{CN}(\mathbf{0}_{N_{j_b} \times 1}, \beta_{j_b}\; \mbox{diag}( \mathbf{\Phi}_{j_b}) ),
\end{equation}

where $k_{k_b} \ll 1$, $\beta_{j_b} \ll 1,\mathbf{\Phi}_{j_b} =  \mbox{Cov}(\bmr_{j_b})$ and $\bmr_{j_b}$ denotes the undistorted received signal for DL user $j_b \in \mathcal{D}_b$. Let $\bmc_b$ and $\bme_b$ denote the transmit and receive LDR noise for FD BS $b \in \mathcal{B}$, respectively, modelled as 

\begin{equation} 
   \bmc_b \sim \mathcal{CN}(\mathbf{0}_{M_b \times 1}, k_b\; \mbox{diag}( \sum_{n_b \in \mathcal{D}_b} \bmW_b \bmV_{n_b}  \bmV_{n_b}^H \bmW_b^H )) ,
 \end{equation}
 \begin{equation} 
    \bme_b \sim \mathcal{CN}({\mathbf{0}_{N_b^{RF} \times 1} , \beta_b\; \mbox{diag}( \mathbf{\Phi}_b)} ),
\end{equation}
with $ k_b \ll 1 $, $\beta_b \ll 1, \mathbf{\Phi}_b = \mbox{Cov}(\bmr_{b})$ and $\bmr_{b}$ denotes the undistorted received signal by FD BS $b \in \mathcal{B}$ after the analog combiner $\bmF_b$. Let $\bmn_{b}$ and $\bmn_{j_b}$ denote the thermal noise for FD BS $b$ and DL user $j_b$, respectively,
modelled as

\begin{equation} 
    \bmn_{b} \sim \mathcal{CN}(\mathbf{0}_{N_{b} \times 1}, \sigma_{b}^2 \bmI), \quad \quad \bmn_{j_b} \sim \mathcal{CN}(\mathbf{0}_{N_{j_b} \times 1}, \sigma_{j_b}^2 \bmI), 
\end{equation} 
with $\sigma_{b}^2$ and $\sigma_{j_b}^2$ denoting the noise variances.
 
\subsection{Channel Modelling}
We assume perfect CSI, which can be achieved based on the compressed sensing-based techniques developed for mmWave, similar to \cite{kuai2019structured}.
Let $\bmH_{j_b} \in \mathbb{C}^{N_{j_b} \times M_{b}} $ and $\bmH_{k_b} \in \mathbb{C}^{N_b \times M_{k_b}}$ denote the channels between DL user $j_b \in \mathcal{D}_b$ and UL user $k_b \in \mathcal{U}_b$ and FD BS $b \in \mathcal{B}$, respectively.
Let $\bmH_{j_b,k_b} \in \mathbb{C}^{N_{j_b} \times M_{k_b}}$ and $\bmH_{j_b,k_c} \in \mathbb{C}^{N_{j_b} \times M_{k_c}}$ denote the in-cell UL cross-interference (CI) channel (generated from opposite transmission directions) between the DL user $j_b \in \mathcal{D}_b$ and UL user $k_b \in \mathcal{U}_b$ and the out-cell UL CI channel response between the DL user $j_b \in \mathcal{D}_b $ and UL user $k_c \in \mathcal{U}_c$, respectively, with $b \neq c$. Let $\bmH_{j_b,c} \in \mathbb{C}^{N_{j_b} \times M_{c}}$ and $\bmH_{b,k_c} \in \mathbb{C}^{N_{b} \times M_{k_c}}$ denote the interference channels responses from FD BS $c \in \mathcal{B}$ to DL user $j_b \in \mathcal{D}_b$ and from UL user $k_c \in \mathcal{U}_c$ to FD BS $b$, respectively, with $ c \neq b$. Let  $\bmH_{b,c} \in \mathbb{C}^{N_b \times M_c} $ and $\bmH_{b,b} \in \mathbb{C}^{N_b \times M_b}$ denote the DL CI channel response from FD BS $c \in \mathcal{B}$ to FD BS $b \in \mathcal{B}$, with $c \neq b$, and the SI channel response for FD BS $b \in \mathcal{B}$, respectively. In mmWave, channel $\bmH_{k_b}$ can be modelled as  
\begin{equation} \label{channel_Model}
\begin{aligned}
     \bmH_{k_b} = \sqrt{\frac{{1}}{{N_{k_b}}}} \sum_{n = 1}^{N_{k_b}^p} & \alpha_{k_b}^{n} \mathbf{a}_{r}^b\;(\phi_{k_b}^n ) \; \mathbf{a}_{t}^{k_b}(\theta_{k_b}^{n})^T,
\end{aligned}
\end{equation} 
where the scalars $N_{k_b}^p$ and $\alpha_{k_b}^{n} $ denote the number of paths and a complex Gaussian random variable with amplitudes and phases distributed according to the Rayleigh and uniform distribution, respectively. The vectors $\mathbf{a}_{r}^b(\phi_{k_b}^n)$ and  $\mathbf{a}_{t}^{k_b}(\theta_{k_b}^{n})^T$ denote the receive and transmit antenna array response for FD BS $b \in \mathcal{B}$ and UL user $k_b \in \mathcal{U}_b$, respectively, with the angle of arrival (AoA) $\phi_{k_b}^{n}$ and angle of departure (AoD) $\theta_{k_b}^{n}$, respectively.
The channel responses $\bmH_{j_b},\bmH_{j_b,k_b}, \bmH_{j_b,k_c},\bmH_{j_b,c}$ and $\bmH_{b,k_c} $ can be modelled similarly as $\eqref{channel_Model}$ and the SI channel $\bmH_{b,b} \in \mathbb{C}^{N_b \times M_b}$ can be modelled as \cite{sheemar2021practical_HYBF}
  
\begin{table}[t] \tiny
\centering   \caption{Main Notations}
    \resizebox{8cm}{!}{%   
    \begin{tabular}{|p{6mm}|p{60mm}|}
        \hline  %\multicolumn{2}{|c|}{} \\
       %\hline           
         \mbox{$\mathbf{U}_{k_b}$} & Digital beamformer for UL user $k_b$ \\ 
         \hline
         \mbox{$\mathbf{V}_{j_b}$} & Digital beamformer for DL user $j_b$ \\ 
          \hline
         \mbox{$\bmW_b$} & Analog beamformer for BS $b$ \\
          \hline
           \mbox{$\bmF_b$} & Analog combiner at BS b\\
          \hline
          \mbox{$\bmH_{j_b}$} & Direct channel for DL user $j_b$ \\
          \hline
           \mbox{$\bmH_{k_b}$} & Direct channel for UL user $k_b$ \\
          \hline
          \mbox{$\bmH_{j_b,k_b}$} & CI channel from $k_b$ to $j_b$ \\
          \hline
          \mbox{$\bmH_{j_b,k_c}$} & CI channel from $k_c$ to $j_b$ \\
           \hline
           \mbox{$\bmH_{j_b,c}$} & Interference channel from BS $c$ to $j_b$ \\
           \hline
           \mbox{$\bmH_{b,k_c}$} & Interference channel from  user $k_c$ to the BS $b$ \\
           \hline
           \mbox{$\bmH_{b,c}$} & BS-to-BS interference channel from BS $c$ to the BS $b$ \\
           \hline
           \mbox{$\bmH_{b,b}$} & SI channel for BS $b$ \\
           \hline
    \end{tabular}}     \vspace{-2mm}
     \label{table_parametri_definiti}
\end{table}
\begin{equation} \label{SI_Channel}
    \bmH_{b} = \sqrt{\frac{\kappa_b}{\kappa_b+1}} \bmH_{b}^L + \sqrt{\frac{1}{\kappa_b+1}} \bmH_{b}^R,
\end{equation}
where $\bmH_{b}^R$ denotes the reflected components which can be modelled as \eqref{channel_Model} and 
$\bmH_{b}^L$ denotes the line of sight (LoS) channel, given as
\begin{equation}
    \bmH_{b}^L(m,n) = \frac{\rho_b}{r_{m,n}} e^{-j 2 \pi \frac{r_{m,n}}{\lambda}}.
\end{equation}
The scalars $\kappa_b$, $\rho_b$, $r_{m,n}$ and $\lambda$ denote the Rician factor, the power normalization constant to assure  $\mathbb{E}(||\bmH_{b}^L(m,n)||_F^2)= M_b N_b$, the distance between $m$-th receive and $n$-th transmit antenna and the wavelength, respectively. The main notations described above are summarized in Table \ref{table_parametri_definiti}.

\subsection{Problem Formulation}

Let $\bmy_{j_b}$ and $\bmy_{k_b}$ denote the signals received by the DL user $j_b \in \mathcal{D}_b$ and by the FD BS $b \in \mathcal{B}$ from UL user $k_b \in \mathcal{U}_b$ after the analog combiner $\bmF_b$, respectively, given as
 
\begin{equation}
     \begin{aligned}
            \bmy_{j_b}  =  & \bmH_{j_b}   ( \hspace{-1mm}  \sum_{n_b \in \mathcal{D}_b} \hspace{-1mm} \bmW_b \bmV_{n_b} \bms_{n_b} \hspace{-1mm}+ \bmc_{b})  + \bme_{j_b} \hspace{-1mm} + \bmn_{j_b}   + \hspace{-1mm} \sum_{k_b \in \mathcal{U}_b} \hspace{-2mm} \bmH_{j_b,k_b} ( \bmU_{k_b}    \bms_{k_b} \\& + \bmc_{k_b} )  +  \sum_{\substack{c \in \mathcal{B}, c \neq b}}  \Big[\bmH_{j_b,c} (\sum_{n_c \in \mathcal{D}_c} \hspace{-2mm} \bmW_c  \bmV_{n_c} \bms_{n_c}  + \bmc_{c} )   \\&+ \sum_{k_c \in \mathcal{U}_c} \bmH_{j_b,k_c} (\bmU_{k_c} \bms_{k_c} + \bmc_{k_c}) \Big],
     \end{aligned} \label{DL_signal}
  \end{equation}
 \begin{equation}
     \begin{aligned}
       \bmy_{k_b} = & {\bmF_b}^H \Big[ \sum_{k_b \in \mathcal{U}_b} \bmH_{k_b} ( \bmU_{k_b} \bms_{k_b} + \bmc_{k_b}) + \bmn_b  + \bmH_{b,b} ( \sum_{j_b \in \mathcal{D}_b} \hspace{-2mm} \bmW_b   \\&  \bmV_{j_b} \bms_{j_b}  + \bmc_{b})   +   \sum_{\substack{c \in \mathcal{B}, c \neq b}} [\bmH_{b,c} ( \sum_{j_c \in \mathcal{D}_c} \bmW_c \bmV_{j_c} \bms_{j_c} + \bmc_{c} )  \\& +   \sum_{k_c \in \mathcal{U}_c} \bmH_{b,k_c} ( \bmU_{k_c} \bms_{k_c} + \bmc_{k_c} ) ] \Big] + \bme_b .
     \end{aligned}   \label{UL_signal}      
  \end{equation}

Let $\overline{k}_b$, $\overline{j}_b$ and $\overline{b}$ denote the indices in the sets $\mathcal{U}_b$, $\mathcal{D}_b$ and $\mathcal{B}$ without the elements $k_b$, $j_b$ and $b$, respectively. Let $\bmT_{k_b} \triangleq \bmU_{k_b} \bmU_{k_b}^H$ and $\bmQ_{j_b} \triangleq \bmW_b \bmV_{j_b} \bmV_{j_b}^H \bmW_b^H$ denote the transmit covariance matrices of UL user $k_b \in \mathcal{U}_b$ and of FD BS $b \in \mathcal{B}$ intended for its DL user $j_b \in \mathcal{D}_b$, respectively. Let ($\bmR_{k_b}$)  $\bmR_{\overline{k}_b}$ and ($\bmR_{j_b}$) $\bmR_{\overline{j}_b}$ denote the (signal plus) interference plus noise covariance matrices received
by the FD BS $b \in \mathcal{B}$ from UL user $k_b \in \mathcal{U}_b$ and by the DL user $j_b \in \mathcal{D}_b$, respectively. The matrices $\bmR_{k_b}$ and $\bmR_{j_b}$ can be written as \eqref{cov}  and $\bmR_{\overline{k}_b}$ and $\bmR_{\overline{j}_b}$ can be obtained as $\bmR_{\overline{k}_b} = \bmR_{k_b} - \bmH_{j_b} \bmQ_{j_b}\bmH_{j_b}^H$ and $\bmR_{\overline{j}_b} = \bmR_{j_b} - {\bmF_b}^H\bmH_{k_b} \bmT_{k_b} \bmH_{k_b}^H {\bmF_b}$, respectively.

\begin{figure*} \small
\begin{subequations} \label{cov}
     \begin{equation} 
     \begin{aligned}
              \bmR_{j_b}  = & \bmH_{j_b} \bmQ_{j_b}\bmH_{j_b}^H  +\bmH_{j_b} ( \hspace{-1mm}\sum_{\substack{n_b \in \mathcal{D}_b\\ n_b\neq j_b}} \hspace{-1mm}\bmQ_{n_b} ) \bmH_{j_b}^H +
            \bmH_{j_b} k_b \mbox{diag}(\sum_{n_b \in \mathcal{D}_b} \hspace{-1.3mm} \bmQ_{n_b} ) \bmH_{j_b}^H   +  \hspace{-1.3mm} \sum_{k_b \in \mathcal{U}_b} \hspace{-1.3mm} \bmH_{j_b,k_b}  ( \bmT_{k_b}  + k_{k_b} \mbox{diag}(\hspace{-0.5mm} \bmT_{k_b}\hspace{-0.5mm})) \bmH_{j_b,k_b}^H \hspace{-1mm}   + \hspace{-1mm} \sum_{\substack{c \in  \mathcal{B}\\c \neq b}}  \bmH_{j_b,c} \\ &(\hspace{-1mm}\sum_{n_c \in \mathcal{D}_c}  \bmQ_{n_c} + k_c \mbox{diag}( \bmQ_{n_c}) ) \bmH_{j_b,c}^H   +\sum_{\substack{c \in \mathcal{B}\\c \neq b}}  \sum_{k_c \in \mathcal{U}_c} \bmH_{j_b,k_c} (\bmT_{k_c} + k_{k_c}  \mbox{diag}(\bmT_{k_c}) )\bmH_{j_b,k_c}^H   + \beta_{j_b} \mbox{diag}( \Phi_{j_b})  + \sigma_{j_b}^2 \bmI_{N_{j_b}}, 
     \end{aligned} \label{Cov_DL} 
  \end{equation}
     \begin{equation}  
     \begin{aligned}  
       \bmR_{k_b} = &  {\bmF_b}^H (\bmH_{k_b} \bmT_{k_b} \bmH_{k_b}^H +\sum_{\substack{m_b \in \mathcal{U}_b \\m_b\neq k_b}}  \bmH_{m_b}  \bmT_{m_b} \bmH_{m_b}^H  + \sum_{m_b \in \mathcal{U}_b} k_{m_b} \bmH_{m_b}  \mbox{diag}(\bmT_{m_b}) \bmH_{m_b}^H  +  \bmH_{b,b} \sum_{j_b \in \mathcal{D}_b} ( \bmQ_{j_b}   + k_{b} \mbox{diag}(\bmQ_{j_b}))\bmH_{b,b}^H   \\& +  \sum_{\substack{c \in \mathcal{B}\\c \neq b}} \bmH_{b,c}  \sum_{j_c \in \mathcal{D}_c} (\bmQ_{j_c}  + k_c \mbox{diag}(\bmQ_{j_c} ) )\bmH_{b,c}^H     +  \sum_{\substack{c \in \mathcal{B}\\c \neq b}} \sum_{k_c \in \mathcal{U}_c} \bmH_{b,k_c} ( \bmT_{k_c}   + k_{k_c} \mbox{diag}(\bmT_{k_c}) )\bmH_{b,k_c}^H  + \sigma_{b}^2 \bmI_{n_b}  ) {\bmF_b} + \beta_b \mbox{diag}(\mathbf{\Phi}_{b}),  
     \end{aligned}   \label{Cov_BS}  
  \end{equation}  \vspace{-2mm}
 \end{subequations} \hrulefill
 \end{figure*}

%The matrices $\mathbf{\Phi}_{j_b}$ and $\mathbf{\Phi}_{b}$ for the receive LDR noise in \eqref{LDR_1} and \eqref{LDR_2} can be recovered from \eqref{Cov_DL} and \eqref{Cov_BS} with
%$\beta_{j_b} = 0$ and $\beta_{b}=0$, respectively.
 
 %\begin{figure*}

 %\end{figure*}

The WSR maximization problem for HYBF in the considered scenario with  $J_b$ DL and $U_b$ UL multi-antenna users $\forall b \in \mathcal{B}$, under the joint sum-power, unit-modulus and discrete phase-shifters constraints can be stated as
 
\begin{subequations} \label{WSR}
\begin{equation}   
\begin{aligned}
      \underset{\substack{\bmU,\bmV\\ \bmW_b, \bmF_b}}{\max}  \hspace{-0.5mm} \sum_{b \in \mathcal{B}} \hspace{-0.5mm}   \Big[ \hspace{-1.5mm}  \sum_{k_b \in \mathcal{U}_b} \hspace{-2mm}  w_{k_b} \mbox{ln det}( \bmR_{\overline{k}_b}^{-1} \bmR_{k_b} ) \hspace{-0.5mm} +  \hspace{-2mm} \sum_{j_b \in \mathcal{D}_b}    \hspace{-1.5mm} w_{j_b} \mbox{ln det}( \bmR_{\overline{j}_b}^{-1} \bmR_{j_b}) \Big] 
\end{aligned}
\end{equation}
\begin{equation}
\text{s.t.} \quad  \mbox{Tr}(\bmU_{k_b} \bmU_{k_b}^H ) \leq p_{k_b},  \quad  \forall k_b \in \mathcal{U}_b,   \end{equation}
\begin{equation}
  \mbox{Tr}(\sum_{j \in \mathcal{D}_b} \bmW_b \bmV_j  \bmV_j^H \bmW_b^H ) \leq  p_b, \quad \forall b \in \mathcal{B}, \label{c2}
\end{equation}  
\begin{equation} 
 \bmW_b(m,n) \;\&\; \bmF_b (m,n) \in  \mathcal{P}_b,  \quad  \;  \forall m,n. \label{c3}
\end{equation}
\end{subequations}
The scalar $w_{i}$ denotes the rate weight for user $i$ and $p_b$ denotes the sum-power constraint for BS $b \in \mathcal{B}$, $\bmU$ and $\bmV$ denote the collection of  digital beamformers in UL and DL, respectively, and $\bmW_b$ and $\bmF_b$ denote the  collections of analog beamformers and combiners, respectively.
%\emph{Remark 1:} The WSR achieved with \eqref{WSR} is not affected with the minimum-mean-squared-error (MMSE) combiners $(4)-(9)$ \cite{christensen2008weighted}. Therefore, they can be omitted during the optimization process and can be chosen as the MMSE receivers after solving \eqref{WSR}. The number of digital combiners would be equal to the number of total users in the multicell FD network. By omitting them, the HYBF design simplifies and the per-iteration computational complexity reduces significantly. Moreover, it will significantly reduce the amount of data exchange required among the FD BSs by the P$\&$D-HYBF.

\section{Minorization-Maximization}\label{simplificazione_problem}
Finding the global optimum of problem \eqref{WSR} is extremely challenging as it is non-concave in the transmit covariance matrices $\bmT_{k_b}$ and $\bmQ_{j_b}$ due to interference. To find its sub-optimal solution based on alternating optimization, we leverage the MM method \cite{stoica2004cyclic}, which allows reformulating \eqref{WSR} with its minorizer using the difference-of-convex (DC) functions \cite{stoica2004cyclic}.

%The optimization problem represented by equation \eqref{WSR} poses significant challenges in achieving global optimality due to the non-concave nature arising from interference terms within the transmit covariance matrices, $\bmT_{k_b}$ and $\bmQ_{j_b}$. To address this, we adopt an approach based on alternating optimization, aiming to obtain a sub-optimal solution. In this pursuit, we employ the minorization-maximization (MM) method \cite{stoica2004cyclic}, which enables us to reformulate equation (WSR) with its minorizer through the utilization of difference-of-convex (DC) functions \cite{stoica2004cyclic}.

Let $\mbox{WR}_{k_b}^{UL}$ and  $\mbox{WR}_{j_b}^{DL}$ denote the weighted rate (WR) of users $k_b \in \mathcal{U}_b$ and $j_b \in \mathcal{D}_b$, respectively, and let $\mbox{WSR}_{\overline{k}_b}^{UL}$  and $\mbox{WSR}_{\overline{j}_b}^{DL}$ denote the WSR of users in UL and DL outside the cell $b$, respectively. The dependence of the global WSR in \eqref{WSR} on the aforementioned terms can
be highlighted as   
 
 \begin{equation}  
 \begin{aligned}
           \mbox{WSR}\hspace{-0.5mm} = \mbox{WR}_{k_b}^{UL} \hspace{-0.5mm}+ \hspace{-0.5mm}\mbox{WSR}_{\overline{k}_b}^{UL}+ \hspace{-0.5mm} \mbox{WR}_{j_b}^{DL}\hspace{-0.5mm}+\hspace{-0.5mm}\mbox{WSR}_{\overline{j}_b}^{DL}\hspace{-0.5mm} + \hspace{-0.5mm}\mbox{WSR}_{\overline{b}}^{UL} \hspace{-0.4mm}+ \hspace{-0.5mm}\mbox{WSR}_{\overline{b}}^{DL}
 \end{aligned}
\label{eqWSR}
\end{equation}
in which the WSR in UL and DL for FD BS $b \in \mathcal{B}$ is given as $\mbox{WSR}_b^{UL} =\mbox{WR}_{k_b}^{UL} \hspace{-0.3mm}+ \hspace{-0.4mm}\mbox{WSR}_{\overline{k}_b}^{UL}$ and $\mbox{WSR}_{b}^{DL} = \mbox{WR}_{j_b}^{DL} + \mbox{WSR}_{\overline{j}_b}^{DL}$, respectively. Considering the dependence of the transmit covariance matrices on the global WSR, only $\mbox{WR}_{k_b}^{UL}$ is concave in $\bmT_{k_b}$ and $\mbox{WSR}_{\overline{k}_b}^{UL}$,
$\mbox{WSR}_{b}^{DL}$, $\mbox{WSR}_{\overline{b}}^{UL}$
and $\mbox{WSR}_{\overline{b}}^{DL}$ are non concave in $\bmT_{k_b}$ due to interference. Similarly, only $\mbox{WR}_{j_b}^{DL}$ is concave in $\bmQ_{j_b}$ and $\mbox{WSR}_{\overline{j}_b}^{DL}$, $\mbox{WSR}_{b}^{UL}$,$\mbox{WSR}_{\overline{b}}^{UL}$,$\mbox{WSR}_{\overline{b}}^{DL}$ are non concave in $\bmQ_{j_b}$. As a linear function is simultaneously convex and concave, DC functions introduce the first order
Taylor series expansion of $\mbox{WSR}_{\overline{k}_b}^{UL}$,
$\mbox{WSR}_{b}^{DL}$, $\mbox{WSR}_{\overline{b}}^{UL}$
and $\mbox{WSR}_{\overline{b}}^{DL}$ in $\bmT_{k_b}$, around $\hat{\bmT}_{k_b}$ (i.e. around all $\bmT_{k_b}$), and for $\mbox{WSR}_{\overline{j}_b}^{DL}$, $\mbox{WSR}_{b}^{UL}$,$\mbox{WSR}_{\overline{b}}^{UL}$,$\mbox{WSR}_{\overline{b}}^{DL}$ around $\hat{\bmQ}_{j_b}$ (i.e. around all $\bmQ_{j_b}$).
Let $\hat{\bmT}$ and $\hat{\bmQ}$ denote the 
sets of all such $\hat{\bmT}_{k_b}$ and $\hat{\bmQ}_{j_b}$, respectively. The tangent expressions for the non-concave terms for $\bmT_{k_b}$ can be written by computing the gradients

 \begin{subequations} \label{grad_kb}
\begin{equation} 
    \hat{\bmG}_{\overline{k}_b,b}^{UL} = - \frac{\partial \mbox{WSR}_{\overline{k}_b}^{UL}}{\partial \bmT_{k_b}}|_{\hat{\mathbf{T}},\hat{\mathbf{Q}}}
    ,\; \hat{\bmG}_{{k_b},b}^{DL}  = - \frac{\partial \mbox{WSR}_b^{DL}}{\partial \bmT_{k_b}}|_{\hat{\mathbf{T}}, \hat{\mathbf{Q}}}, 
     \end{equation}
    \begin{equation}
    \hat{\bmG}_{k_b,\overline{b}}^{UL} = - \frac{\partial \mbox{WSR}_{\overline{b}}^{UL}}{\partial \bmT_{k_b}} 
|_{\hat{\mathbf{T}},\hat{\mathbf{Q}}}, \; \hat{\bmG}_{{k_b},\overline{b}}^{DL}  = - \frac{\partial \mbox{WSR}_{\overline{b}}^{DL}}{\partial \bmT_{k_b}}|_{\hat{\mathbf{T}},\hat{\mathbf{Q}}}, 
\end{equation}
  \end{subequations}
which allow to write the minorizers, denoted as $\underline{\mbox{WSR}}_{\overline{k}_b}^{UL},$ $\underline{\mbox{WSR}}_b^{DL},$ $\underline{\mbox{WSR}}_{\overline{b}}^{UL}$ and $\underline{\mbox{WSR}}_{\overline{b}}^{DL}$ with respect to $\bmT_{k_b}$. Similarly, for the transmit covariance matrix $\bmQ_{j_b}$, we have the gradients

 \begin{subequations} \label{grad_jb}
\begin{equation} 
    \hat{\bmG}_{j_b,b}^{UL} = - \frac{\partial \mbox{WSR}_b^{UL}}{\partial \bmQ_{j_b}}|_{\hat{\mathbf{T}},\hat{\mathbf{Q}}}
    , \;\hat{\bmG}_{\overline{j}_b,b}^{DL}  = - \frac{\partial \mbox{WSR}_{\overline{j}_b}^{DL}}{\partial \bmQ_{j_b}}|_{\hat{\mathbf{T}},\hat{\mathbf{Q}}},
    \end{equation}
    \begin{equation}
    \hat{\bmG}_{j_b,\overline{b}}^{UL} = - \frac{\partial \mbox{WSR}_{\overline{b}}^{UL}}{\partial \bmQ_{j_b}}|_{\hat{\mathbf{T}},\hat{\mathbf{Q}}}
    ,  \;\hat{\bmG}_{{j_b},\overline{b}}^{DL}  = - \frac{\partial \mbox{WSR}_{\overline{b}}^{DL}}{\partial \bmQ_{j_b}}|_{\hat{\mathbf{T}},\hat{\mathbf{Q}}}, 
\end{equation}
  \end{subequations}  
which allow to write the minorizers $\underline{\mbox{WSR}}_b^{UL},$ $\underline{\mbox{WSR}}_{\overline{j}_b}^{DL},$ $\underline{\mbox{WSR}}_{\overline{b}}^{UL} $ and $\underline{\mbox{WSR}}_{\overline{b}}^{DL}$ with respect to $\bmQ_{j_b}$. The gradients \eqref{grad_kb} and \eqref{grad_jb} can be computed by applying the matrix differentiation properties and they are reported in Table \ref{table_gradients}.

%Hence, the DC functions constitute a MM approach, regardless of the restatement of the transmit covariance matrices $\bmT_{k_b}$ and $\bmQ_{j_b}$ as a function of the beamformers. 

\begin{table*}  
\centering
    \caption{Gradients expressions to construct the minorized WSR cost function.}
    \resizebox{18cm}{!}{%
    \begin{tabular}{|p{9mm}|p{155mm}|}
       \hline
      $\hat{\bmG}_{\overline{k}_b,b}^{UL}$  & $\sum_{\substack{m_b \in \mathcal{U}_b\\m_b \neq k_b}} w_{m_b} [\bmH_{m_b}^H \bmF_b (\bmR_{\overline{m}_b}^{-1}  - \bmR_{m_b}^{-1}  + \beta_{b} \mbox{diag}(\bmR_{\overline{m}_b}^{-1} - \bmR_{m_b}^{-1} ) ) {\bmF_b}^H \bmH_{m_b} +  k_{m_b} \mbox{diag}(\bmH_{m_b}^H \bmF_b (\bmR_{\overline{m}_b}^{-1} - \bmR_{m_b}^{-1}) {\bmF_b}^H \bmH_{m_b}) ].$  \\ \hline
       $\hat{\bmG}_{{k_b},b}^{DL}$ & $\sum_{\substack{j_b \in \mathcal{D}_b}} w_{j_b} [\bmH_{j_b,k_b}^H (\bmR_{\overline{j}_b}^{-1}  - \bmR_{j_b}^{-1}  + \beta_{j_b} \mbox{diag}(\bmR_{\overline{j}_b}^{-1} -  \bmR_{j_b}^{-1} ) ) \bmH_{j_b,k_c} +  k_{k_b} \mbox{diag}(\bmH_{j_b,k_b}^H (\bmR_{\overline{j}_b}^{-1} - \bmR_{j_b}^{-1}) \bmH_{j_b,k_b}) ]$  \\ \hline
      $\hat{\bmG}_{{k_b},\overline{b}}^{UL}$ & $ \sum_{\substack{c \in \mathcal{B}\\c\neq b}} \sum_{\substack{k_c \in \mathcal{U}_c}} w_{k_c} [\bmH_{c,k_b}^H {\bmF_c} (\bmR_{\overline{k}_c}^{-1}  - \bmR_{k_c}^{-1}  + \beta_{c} \mbox{diag}(\bmR_{\overline{k}_c}^{-1} -  \bmR_{k_c}^{-1} ) ) {\bmF_c}^H \bmH_{c,k_b} +  k_{k_b} \mbox{diag}(\bmH_{c,k_b}^H \bmF_c (\bmR_{\overline{k}_c}^{-1} - \bmR_{k_c}^{-1}) {\bmF_c}^H \bmH_{c,k_b}) ]$  \\ \hline
       $\hat{\bmG}_{{k_b},\overline{b}}^{DL}$ & $\sum_{\substack{c \in \mathcal{B}\\c\neq b}} \sum_{\substack{j_c \in \mathcal{D}_c}} w_{j_c} [\bmH_{j_c,k_b}^H  (\bmR_{\overline{j}_c}^{-1} - \bmR_{j_c}^{-1}  + \beta_{j_c} \mbox{diag}(\bmR_{\overline{j}_c}^{-1} -  \bmR_{j_c}^{-1} ) )  \bmH_{c,k_b} +  k_{k_b} \mbox{diag}(\bmH_{c,k_b}^H  (\bmR_{\overline{j}_c}^{-1} - \bmR_{j_c}^{-1})   \bmH_{c,k_b}) ]$ \\ \hline
      $ \hat{\bmG}_{{j_b},b}^{UL}$ & $\sum_{\substack{k_b \in \mathcal{U}_b}} w_{k_b} [\bmH_{b,b}^H \bmF_b(\bmR_{\overline{k}_b}^{-1}  - \bmR_{k_b}^{-1}  + \beta_{b} \mbox{diag}(\bmR_{\overline{k}_b}^{-1} -  \bmR_{k_b}^{-1} ) ) {\bmF_b}^H \bmH_{b,b} +  k_{k_b} \mbox{diag}(\bmH_{b,b}^H  {\bmF_b} (\bmR_{\overline{k}_b}^{-1} - \bmR_{k_b}^{-1})  {\bmF_b}^H \bmH_{b,b} ) ]$ \\ \hline
      $\hat{\bmG}_{\overline{j}_b,b}^{DL}$ & $ \sum_{\substack{l_b \in \mathcal{D}_b\\ l_b \neq j_b}} w_{l_b} [ \bmH_{l_b}^H (\bmR_{\overline{l}_b}^{-1}  - \bmR_{l_b}^{-1} + \beta_{l_b} \mbox{diag}(\bmR_{\overline{l}_b}^{-1} -  \bmR_{l_b}^{-1} )  ) \bmH_{l_b}^H + k_c \mbox{diag}(\bmH_{l_b}^H (\bmR_{\overline{l}_b}^{-1}  - \bmR_{l_b}^{-1}) \bmH_{l_b}) ]$ \\ \hline
      $\hat{\bmG}_{j_b,\overline{b}}^{UL}$ & $\sum_{\substack{c \in \mathcal{B}\\c \neq b}} \sum_{\substack{k_c \in \mathcal{U}_c}} w_{k_c} [ \bmH_{c,b}^{H} {\bmF_c} (\bmR_{\overline{k}_c}^{-1} - \bmR_{k_c}^{-1} + \beta_c \mbox{diag} (\bmR_{\overline{k}_c}^{-1} - \bmR_{k_c}^{-1} )) {\bmF_c}^H \bmH_{c,b} + k_b \mbox{diag}(\bmH_{c,b}^{H} {\bmF_c} (\bmR_{\overline{k}_c}^{-1} - \bmR_{k_c}^{-1}) {\bmF_c}^H  \bmH_{c,b} ) ]$ \\ \hline
      $\hat{\bmG}_{j_b,\overline{b}}^{DL}$ & $\sum_{\substack{c \in \mathcal{B}\\c \neq b}} \sum_{\substack{j_c \in \mathcal{D}_c}} w_{j_c} [ \bmH_{j_c,b}^{H} (\bmR_{\overline{j}_c}^{-1} - \bmR_{j_c}^{-1} + \beta_{j_c} \mbox{diag} (\bmR_{\overline{j}_c}^{-1} - \bmR_{j_c}^{-1} )) \bmH_{j_c,b} + k_b \mbox{diag}(\bmH_{j_c,b}^{H} (\bmR_{\overline{j}_c}^{-1} - \bmR_{j_c}^{-1})\bmH_{j_c,b}  ]$  \\ \hline
    \end{tabular}} \label{table_gradients} 
\end{table*}

Let $\lambda_{k_b}$ and $\psi_b$ denote the Lagrange multipliers associated with the sum-power constraint for UL user $k_b \in \mathcal{U}_b$ and FD BS $b \in \mathcal{B}$, respectively. For notational convenience, let 
\begin{subequations} \label{sigma_local_var}
\begin{equation}
 \mathbf{\Sigma}_{k_b}^{1} =  \bmH_{k_b}^H \bmF_b \bmR_{\overline{k}_b}^{-1} {\bmF_b}^H \bmH_{k_b}, \quad \mathbf{\Sigma}_{j_b}^{1} =   \bmH_{j_b}^H \bmR_{\overline{j}_b}^{-1} \bmH_{j_b},
  \end{equation}
  \begin{equation}
       \mathbf{\Sigma}_{k_b}^{2}  = \hat{\bmG}_{\overline{k}_b,b}^{UL}+ \hat{\bmG}_{{k_b},b}^{DL}+\hat{\bmG}_{k_b,\overline{b}}^{UL} +\hat{\bmG}_{k_b,\overline{b}}^{DL} + \lambda_{k_b} \bmI,  
\end{equation}\label{diagonal_matrices}
\begin{equation}
    \mathbf{\Sigma}_{j_b}^{2} = \hat{\bmG}_{j_b,b}^{UL}+ \hat{\bmG}_{\overline{j}_b,b}^{DL} +\hat{\bmG}_{j_b,\overline{b}}^{UL}+\hat{\bmG}_{j_b,\overline{b}}^{DL} + \psi_b \bmI.
\end{equation}
\end{subequations} 
We remark that the effect of the LDR noise is captured in the gradients. By considering the minorized WSR constructed with the gradients \eqref{grad_kb}-\eqref{grad_jb}, ignoring the constant terms and the unit-modulus and quantization constraints \eqref{c3}, and augmenting it with the constraints leads to the Lagrangian  
 
\begin{equation}\label{Largrangian_min}
\begin{aligned}
          \mathcal{L} = &\hspace{-1mm}\sum_{b \in \mathcal{B}}   \hspace{-1mm} \big[ \sum_{k_b \in \mathcal{U}_b} (w_{k_b} \mbox{ln det}( \bmI + \bmU_{k_b}^H \mathbf{\Sigma}_{k_b}^{1} \bmU_{k_b})
        - \mbox{Tr}(\bmU_{k_b}^H \mathbf{\Sigma}_{k_b}^{2} \bmU_{k_b})   \\& + \lambda_{k_b} p_{k_b}  )
           +\hspace{-1mm} \sum_{j_b \in \mathcal{D}_b} ( w_{j_b} \mbox{ln det}( \bmI +  \bmV_{j_b}^H {\bmW_b}^H \mathbf{\Sigma}_{j_b}^{1} {\bmW_b} \bmV_{j_b})   \\& -\mbox{Tr}(\bmV_{j_b}^H {\bmW_b}^H  \mathbf{\Sigma}_{j_b}^{2} {\bmW_b} \bmV_{j_b}) +  \psi_b p_b)   \big].
\end{aligned}
\end{equation}

\emph{Remark 1:} The tangent expressions constitute a touching lower bound, and the original WSR and its minorized version
have the same Karush–Kuhn–Tucker (KKT) conditions. Hence any (sub) optimal solution for minorized WSR is also (sub)
optimal for the original WSR

%Let $\lambda_{k_b}$ and $\psi_b$ denote the Lagrange multipliers associated with the sum-power constraint for UL user $k_b \in \mathcal{U}_b$ and BS $b \in \mathcal{B}$, respectively. Augmenting the WSR function \eqref{WSR_convex} with the sum-power constraints yield the Lagrangian \eqref{Largrangian}. Note that \eqref{Largrangian}
%does not consider the quantization constraints on the analog beamformers and combiners, which will be incorporated later.

\section{Centralized Hybrid Beamforming}\label{Centralized_HYBF}
This section presents a novel C-HYBF design based on alternating optimization to solve the WSR maximization problem to a local optimum. Hereafter,
while optimizing one variable, we assume the remaining ones to be fixed and their information to be summarized in the gradients, which are updated at each iteration.

\subsection{Digital Beamforming}
To optimize the digital beamformers
$\bmU_{k_b}$ and $\bmV_{j_b}$ we take the derivatives of \eqref{Largrangian_min}, which leads to the KKT conditions

\begin{subequations}
\begin{equation}
\begin{aligned}
    & \mathbf{\Sigma}_{k_b}^{1} \bmU_{k_b} ( \bmI + \bmU_{k_b}^H \mathbf{\Sigma}_{k_b}^{1} \bmU_{k_b} )^{-1}   - \mathbf{\Sigma}_{k_b}^{2}  \bmU_{k_b} = 0, 
\end{aligned} \label{kkt_uplink} 
\end{equation}   
\begin{equation}
\begin{aligned}
{\bmW_b}^H \mathbf{\Sigma}_{j_b}^{1} {\bmW_b}   \bmV_{j_b} ( \bmI + &  \bmV_{j_b}^H {\bmW_b}^H \mathbf{\Sigma}_{j_b}^{1} {\bmW_b}  \bmV_{j_b} )^{-1}  \\& - {\bmW_b}^H \mathbf{\Sigma}_{j_b}^{2} {\bmW_b}  \bmV_{j_b}= 0.
\end{aligned} \label{kkt_downlink} 
\end{equation}
\end{subequations}
%Given the structure of the KKT conditions \eqref{kkt_uplink}-\eqref{kkt_downlink}, the digital beamformers can be optimized based on the result provided in the following Theorem.
 
\begin{thm}\label{digital_BF_solution}
The optimal digital beamformers $\bmU_{k_b}$ and $\bmV_{j_b}$ for \eqref{Largrangian_min}, given the remaining variables fixed, can be computed as the GDE solution of the pair of the following matrices 
\begin{subequations}
\begin{equation} \label{digital_solution_DL}
    \bmU_{k_b} = \bmD_{d_{k_b}}( \mathbf{\Sigma}_{k_b}^{1}, \mathbf{\Sigma}_{k_b}^{2}  ),
\end{equation}
\begin{equation}   \label{digital_solution_UL}
\bmV_{j_b} = \bmD_{d_{j_b}}( {\bmW_b}^H \mathbf{\Sigma}_{j_b}^{1} {\bmW_b} , {\bmW_b}^H \mathbf{\Sigma}_{j_b}^{2} {\bmW_b}  ),
\end{equation}
\end{subequations}
where the matrix $\bmD_{d_{k_b}}$ ($\bmD_{d_{k_b}}$) selects $d_{k_b} (d_{k_b})$ GDEs.
\end{thm}
\begin{proof}
The proof is provided in Appendix \ref{AP.1}.
\end{proof}
It is imperative to acknowledge that GDEs furnish optimized beamforming directions while leaving the power aspect unaddressed. In our approach, we propose the scaling of beamformers to possess unit-norm columns, thereby facilitating the inclusion of optimal power allocation. It is important to highlight that this scaling operation does not compromise the optimality of the beamforming directions.

%We remark that  problem is stated w.r.t to the non unit-norm beamformers, which can be equivalently stated with respect to the unit-norm beamformers and power matrices, included later in Section \ref{pow_alloc}.

\subsection{Analog Beamforming}

Consider first the optimization of unconstrained analog beamformer $\bmW_b, 
 \forall b \in \mathcal{B}$, without \eqref{c3}, for which the optimization problem can be written as

\begin{equation} \label{analog_BF_resated}
\begin{aligned}
  \underset{\substack{\bmW_b}}{\max}   \sum_{j_b \in \mathcal{D}_b}  & [w_{j_b} \mbox{ln det}(\bmI + \bmV_{j_b}^H {\bmW_b}^H \mathbf{\Sigma}_{j_b}^{1} \bmW_b \bmV_{j_b} )  \\& -\mbox{Tr}(\bmV_{j_b}^H {\bmW_b}^H \mathbf{\Sigma}_{j_b}^{2}  \bmW_b \bmV_{j_b} )].
\end{aligned}
\end{equation}
By taking its derivative we get to the following KKT condition 
\begin{equation} \label{kkt_analog_beamformer}
     \begin{aligned}
    \sum_{j_b \in \mathcal{D}}   (\mathbf{\Sigma}_{j_b}^{1}\bmW_b \mathbf{V}_{j_b} \mathbf{V}_{j_b}^H (\bmI + \mathbf{V}_{j_b} & \mathbf{V}_{j_b}^H {\bmW_b}^H \mathbf{\Sigma}_{j_b}^{1}  \bmW_b  )^{-1} \\& -  \mathbf{\Sigma}_{j_b}^{2} \bmW_b \mathbf{V}_{j_b} \mathbf{V}_{j_b}^H ) =0.
    \end{aligned}
\end{equation}
 
\begin{thm}\label{analog_BF_solution}
The optimal vectorized unconstrained analog beamformer $\bmW_b$ for \eqref{analog_BF_resated} can be optimized as one GDE solution of the pair of the sum of following matrices
\begin{equation} \label{analog_BF}
    \begin{aligned}
      \mbox{vec}(\bmW_b)  = &  \bmD_{1}(\hspace{-1mm} \sum_{j_b \in \mathcal{D}_b} \hspace{-1mm}(\mathbf{V}_{j_b} \mathbf{V}_{j_b}^H (\bmI + \mathbf{V}_{j_b} \mathbf{V}_{j_b}^H {\bmW_b}^H \mathbf{\Sigma}_{j_b}^{1}\bmW_b )^{-1})^T  \\& \otimes \mathbf{\Sigma}_{j_b}^{1},\;  \sum_{j_b \in \mathcal{D}_b} (\mathbf{V}_{j_b} \mathbf{V}_{j_b}^H )^T   \otimes (\mathbf{\Sigma}_{j_b}^{2})).
    \end{aligned}  
\end{equation} 
\begin{proof}
The proof is provided in Appendix \ref{AP.2}
\end{proof} % \vspace{-2mm}
 \end{thm}  
The result stated in Theorem \ref{analog_BF_solution} provides the optimized vectorized unconstrained analog beamformer. Operation $\mbox{unvec}(\mbox{vec}(\bmW_b))$ is required to reshape it into the correct dimensions. To projects its elements on the set $\mathcal{P}_b$,
we do $\bmW_b = \mathbb{Q}_b(\angle\bmW_b(m,n)) \in \mathcal{P}_b, \forall m,n$. Note that such operation results in optimality loss, which  depends on the resolution of the phase shifts, and will be evaluated in Section \ref{simulazioni}.

 \subsection{Analog Combining}
 %Consider the optimization of the analog combiner $\bmF_b$, which is straightforward compared to the analog beamformer. Note that the analog combiners do not appear in the trace operators of \eqref{Largrangian_min} but only in $\mathbf{\Sigma}_{k_b}^{1}, \forall k_b$ in cell $b\in \mathcal{B}$ as they do not generate any interference towards other links. Therefore, to optimize $\bmF_b$ we can directly consider the original problem \eqref{WSR} which is purely concave for $\bmF_b$.

 The optimization process for the analog combiner, denoted as $\bmF_b$, is comparatively straightforward when compared to the optimization of the analog beamformer. It is noteworthy that the analog combiners do not appear in the trace operators of equation \eqref{Largrangian_min}. Instead, they are solely present in $\mathbf{\Sigma}_{k_b}^{1}, \forall k_b$ within cell $b \in \mathcal{B}$. This absence is due to the fact that the analog combiners do not generate any interference towards other links. Consequently, in order to optimize $\bmF_b$, we can directly focus on the original problem instead of minorized WSR with respect to $\bmF_b$.
 
 %The analog combiner must combine the signal received at the antennas of the FD BSs. 
 %By considering the unconstrained $\bmF_b$ and using the properties of logarithm function, the WSR maximization problem with respect to $\bmF_b$ can be stated as

%For the optimization of $\bmF_b$, \eqref{WSR} is purely concave and can be solved directly.
%Given that the analog combiner $\bmF_b$ does not generate any interference, the WSR is purely concave with respect to $\bmF_b$ in the received covariance matrices $\bmR_{k_b}$ and $\bmR_{\overline{k}_b}$ for $b \in \mathcal{B}$. Therefore, the original WSR maximization problem \eqref{WSR} can be considered to optimize $\bmF_b$.

%The objective of the analog combiner $\bmF_b$ is to combine the received covariance matrices at the antenna level such that the WSR is maximized.
%Let $(\bmR_{k_b}^a)$ $\bmR_{\overline{k}_b}^a$ denote the (signal plus) interference and noise covariance matrices received at the antennas of the FD BS $b \in \mathcal{B}$ to be combined with $\bmF_b$. Given $\bmR_{k_b}^a$ and $\bmR_{\overline{k}_b}^a$, the matrices $\bmR_{k_b}$ and $\bmR_{\overline{k}_b}$ can be recovered as $\bmR_{k_b} = {\bmF_b}^H \bmR_{k_b}^a \bmF_b$ and $\bmR_{\overline{k}_b} = {\bmF_b}^H \bmR_{\overline{k}_b}^a \bmF_b$. Problem \eqref{WSR} to optimize $\bmF_b$, by using the properties of the logarithm function, can be restated as

The primary objective of the analog combiner is to effectively combine the received covariance matrices at the antenna level in order to maximize the WSR. Within this context, let $(\bmR_{k_b}^a)$ and $\bmR_{\overline{k}b}^a$ represent the (signal plus) interference and noise covariance matrices received at the antennas of the FD BS $b \in \mathcal{B}$, which are to be combined with $\bmF_b$. By leveraging $\bmR_{k_b}^a$ and $\bmR_{\overline{k}b}^a$, the matrices $\bmR_{k_b}$ and $\bmR_{\overline{k}b}$ can be derived as $\bmR{k_b} = {\bmF_b}^H \bmR_{k_b}^a \bmF_b$ and $\bmR_{\overline{k}b} = {\bmF_b}^H \bmR{\overline{k}_b}^a \bmF_b$. To optimize the unconstrained $\bmF_b$, the logarithm function's properties can be effectively employed, which allows stating the following optimization problem

\begin{equation} \label{combiner_opt}
\begin{aligned}
   \underset{\substack{\bmF_b}}{\max}   \sum_{k_b \in \mathcal{U}_b} \hspace{-2mm} w_{k_b}  [ & \mbox{ln det}( {\bmF_b}^H \bmR_{k_b}^a  \bmF_b  ) - \mbox{ln det}({\bmF_b}^H \bmR_{\overline{k}_b}^a \bmF_b ) ].
\end{aligned}
\end{equation}
By solving the purely concave optimization problem \eqref{combiner_opt}, we obtain the optimal analog combiner as follows.

%Taking the derivative of \eqref{combiner_opt} with respect to $\bmF_b$ leads to the following KKT condition
%\begin{equation} \label{KKT_combiner}
%\begin{aligned}
 %   \sum_{k_b \in \mathcal{U}_b} & w_{k_b} {\bmR_{k_b}^a}  {\bmF_b}  ( {\bmF_b}^H \bmR_{{\bmR_{\overline{k}_b}}}^a  \bmF_b  )^{-1} \\& -  \sum_{k_b \in \mathcal{U}_b} w_{k_b} {\bmR_{\overline{k}_b}^a}  {\bmF_b}  ( {\bmF_b}^H \bmR_{{\bmR_{\overline{k}_b}}}^a  \bmF_b  )^{-1}  =0.
%\end{aligned}
%\end{equation}
%From \eqref{KKT_combiner}, it is immediate to see that the WSR maximizing analog combiner $\bmF_b$ can be obtained as the GDE solution of the pair of the sum of the received covariance matrices at the antenna level from the UL users in the same cell, i.e. 
 \begin{equation} \label{analog_combiner}
     \bmF_b = \bmD_{N_b^{RF}}( \sum_{k_b \in \mathcal{U}_b} w_{k_b} \bmR_{k_b}^a ,\sum_{k_b \in \mathcal{U}_b} w_{k_b} \bmR_{\overline{k}_b}^a ),
 \end{equation}
 where the matrix $\bmD_{N_b^{RF}}$ selects GDEs equal to the number of receive RF chains $N_b^{RF}$ at the FD BS $b \in \mathcal{B}$. In order to satisfy the constraints for \eqref{analog_combiner}, we normalize the amplitudes of $\bmF_b$ using the $\angle \cdot$ operation, and then subject it to quantization through the function $\mathbb{Q}_b(\angle \bmF_b(m,n)) \in \mathcal{P}_b$. It should be noted that this quantization process introduces a loss in optimality, which is contingent upon the resolution of the phase shift and will be assessed in Section \ref{simulazioni}.
 
 %To meet the constraints for \eqref{analog_combiner}, we normalize its amplitudes with the $\angle \cdot $ and pass it through the quantizer such that $\bmF_b = \mathbb{Q}_b(\angle \bmF_b(m,n)) \in \mathcal{P}_b$. Recall that this results in optimality loss, which depends on the phase shift resolution and will be evaluated in Section \ref{simulazioni}.

\subsection{Optimal Power Allocation} \label{pow_alloc}

Let $\mathbf{P}_{k_b}$ and $\mathbf{P}_{j_b}$ denote the diagonal stream power matrices for the UL user $k_b \in \mathcal{U}_{b}$ and DL user $j_b \in \mathcal{D}_{b}$, respectively, to be included in the beamformers $\bmU_{k_b}$ and $\bmV_{k_b}$, respectively. Given the normalized digital beamformers, the optimal power allocation problems can be formally stated as
 
 \begin{subequations}
\begin{equation} \label{power_UL}
\begin{aligned}
    \underset{\mathbf{P}_{k_b}}{\max} &\quad w_{k_b} \mbox{ln det}( \bmI + \bmU_{k_b}^H \mathbf{\Sigma}_{k_b}^{1} \bmU_{k_b}  \mathbf{P}_{k_b} ) - \mbox{Tr}(\bmU_{k_b}^H \mathbf{\Sigma}_{k_b}^{2} \bmU_{k_b}  \mathbf{P}_{k_b} ),   
\end{aligned}
\end{equation}
\begin{equation}\label{power_DL}
\begin{aligned}
        \underset{\mathbf{P}_{j_b}}{\max} \quad  w_{j_b} \mbox{ln det}( \bmI + &\bmV_{j_b}^H {\bmW_b}^H \mathbf{\Sigma}_{j_b}^{1} {\bmW_b} \bmV_{j_b} \mathbf{P}_{j_b} )   \\& - \mbox{Tr}(\bmV_{j_b}^H {\bmW_b}^H \mathbf{\Sigma}_{j_b}^{2} {\bmW_b} \bmV_{j_b} \mathbf{P}_{j_b} ).
\end{aligned}
\end{equation}
 \end{subequations}
%Recall that right multiplying the unit-norm beamformers with powers is equivalent to not normalized digital beamformers which also includes the optimal power allocation, therefore the optimality of \eqref{digital_solution_DL}-\eqref{digital_solution_UL} is preserved.
It is important to remember that when we multiply the unit-norm beamformers with the corresponding powers, it is equivalent to using non-normalized digital beamformers that incorporate the optimal power allocation. Hence, the optimality of \eqref{digital_solution_DL}-\eqref{digital_solution_UL} remains intact. Solving \eqref{power_UL}-\eqref{power_DL} leads to the following optimal power allocation  
  \begin{subequations} \label{optimal_pow}
 \begin{equation} 
     \mathbf{P}_{k_b} \hspace{-1mm}= \hspace{-1mm}( w_{k_b}  (\bmU_{k_b}^H \mathbf{\Sigma}_{k_b}^{2} \bmU_{k_b} )^{-1}\hspace{-1mm}-\hspace{-1mm} (\bmU_{k_b}^H \mathbf{\Sigma}_{k_b}^{1} \bmU_{k_b}  )^{-1} )^{+}, 
       \end{equation}
     \begin{equation}
     \begin{aligned}
       \mathbf{P}_{j_b}  \hspace{-1.5mm} =  \hspace{-1.2mm}( w_{j_b}  ( \bmV_{j_b}^H {\bmW_b}^H \mathbf{\Sigma}_{j_b}^{2} {\bmW_b} \bmV_{j_b})^{-1}  \hspace{-1.7mm}-  (\bmV_{j_b}^H {\bmW_b}^H \mathbf{\Sigma}_{j_b}^{1} {\bmW_b} \bmV_{j_b})^{-1} \hspace{-0.5mm} )^{+}, 
     \end{aligned}
 \end{equation} 
\end{subequations}
where $(\bmX)^+ = max\{\mathbf{0},\bmX\}$. 
Given the optimal stream powers, we can search for the Lagrange multipliers satisfying the total sum-power constraint. Let $\mathbf{P}^{DL}$ and $\mathbf{P}^{UL}$ denote the collection of powers in DL and UL, respectively, and
let $\mathbf{\Lambda}$ and $\mathbf{\Psi}$ denote
the collection of multipliers for $\lambda_{k_b}$ and $\psi_b$, respectively. Given $\eqref{optimal_pow}$, consider the dependence of the Lagrangian only on the multipliers and powers as $\mathcal{L}(\mathbf{\Lambda},\mathbf{\Psi},\mathbf{P}^{DL},\mathbf{P}^{UL})$, obtained by including the power matrices  $\mathbf{P}_{k_b}$ and $\mathbf{P}_{j_b}$ in \eqref{Largrangian_min}.

The multipliers in $\mathbf{\Lambda}$ and $\mathbf{\Psi}$ should be such that the Lagrangian is finite and the values of multipliers are strictly positive, i.e.,
\begin{equation} \label{lag_min_max}
\begin{aligned}
   \underset{\mathbf{\Psi},\mathbf{\Lambda}}{\min}\,\,& \underset{\mathbf{P}^{DL},\mathbf{P}^{UL}}{\max}  \quad  \mathcal{L}(\mathbf{\Lambda},\mathbf{\Psi},\mathbf{P}^{DL},\mathbf{P}^{UL}), \\
  &\quad \quad \mbox{s.t.} \quad \quad \mathbf{\Psi} ,\mathbf{\Lambda} \succeq 0.
\end{aligned}
\end{equation}
The dual function $
    \max_{\mathbf{P}^{DL},\mathbf{P}^{UL}} \mathcal{L}$
is the pointwise supremum of a family of functions of $\mathbf{\Psi},\mathbf{\Lambda}$, it is convex \cite{boyd2004convex} and the
globally optimal values for $\mathbf{\Psi}$ and $\mathbf{\Lambda}$ can be found by using any of the numerous convex-optimization techniques. In this work, we adopt the Bisection method. Let $\underline{\lambda_{k_b}}, \underline{\psi_b}$ and $\overline{\psi_b},\overline{\lambda_{k_b}}$
denote the upper and lower bounds for searching the multipliers $\psi_b$ and $\lambda_{k_b}$, respectively, and let $[0,\lambda_{k_b}^{max}]$ and $[0,\psi_{b}^{max}]$ denote their search range.  
Note also that as the GDE solution is computed given fixed multipliers, doing water-filling for the powers while searching for the multipliers leads to non diagonal power matrices. Hence, consider a SVD of the powers as $[\bmL_{\bmP_i}^{svd}, \bmD_{\bmP_i}^{svd} ,\bmR_{\bmP_i}^{svd}] = svd(\mathbf{P}_{i})$,
with $\bmP_i = \bmP_{k_b}$ or $\bmP_i = \bmP_{j_b}$, and the matrices $\bmL_{\bmP_i}^{svd},\bmD_{\bmP_i}^{svd}$ and $\bmR_{\bmP_i}^{svd}$ denote the 
left unitary, diagonal and right unitary matrices obtained from the SVD of $\bmP_i$. The diagonal structure of the power matrices while searching for the multipliers can be re-established as $\mathbf{P}_{i} = \bmD_{i}$, with $i\in \mathcal{U}_b$ or $ \mathcal{D}_b$.

By using the closed-form expressions derived above, the complete alternating optimization-based C-HYBF procedure to optimize \eqref{Largrangian_min} is formally stated in Algorithm \ref{alg_1}.  

\begin{algorithm}[t]  \footnotesize
\caption{Centralized Hybrid Beamforming}\label{alg_1}
\textbf{Given:} $\mbox{The CSI and rate weights.}$\\
\textbf{Initialize:}\;$\bmW_b, \bmF_b,\bmV_{j_b}, \bmU_{k_b}, \quad \forall j_b \;\&\; \forall k_b $.\\
\textbf{Set:} $\underline{\lambda_{k_b}} = 0,\overline{\lambda_{k_b}} = \lambda_{k_b}^{max}, \underline{\psi_b}=0,\overline{\psi_b}= \psi_b^{max}$, $\forall k_b \;\&\; \forall b $ \\
\textbf{Repeat until convergence}
\begin{algorithmic}
\STATE \hspace{0.001cm} \textbf{for} $b = 1:B$
\STATE \hspace{0.2cm} Compute $\bmW_b$ with \eqref{analog_BF}, do $\mbox{unvec}(\bmW_b)$ and get $\angle \bmW_b$ \\
\STATE \hspace{0.2cm} \textbf{for:} $j_b =1:D_b$ 
\STATE \hspace{0.5cm} Compute $\hat{\bmG}_{j_b,b}^{UL}, \hat{\bmG}_{{j_b},b}^{DL},\hat{\bmG}_{j_b,\overline{b}}^{UL},\hat{\bmG}_{j_b,\overline{b}}^{DL}$ from Table \ref{table_gradients} \\  \hspace{0.5cm} Optimize $\bmV_{j_b}$ with \eqref{digital_solution_DL} and normalize its columns \\
\STATE \hspace{0.2cm} Next $j_b$
\STATE \hspace{0.2cm} \textbf{Repeat until convergence}\\
\STATE \hspace{0.5cm} set $\psi_b = (\underline{\psi_b} + \overline{\psi_b})/2$  \\
\STATE \hspace{0.5cm} \textbf{for} $j_b =1:D_b$
\STATE \hspace{0.9cm} Compute $\mathbf{P}_{j_b}$ with \eqref{optimal_pow}, do SVD, set $\mathbf{P}_{j_b} = \bmD_{P_{j_b}}^{svd}$ \\ \hspace{0.85cm} Set $\bmQ_{j_b} = \bmW_b \bmV_{j_b} \mathbf{P}_{j_b} \bmV_{j_b}^H \bmW_b^H $ 
\STATE \hspace{0.5cm} Next $j_b$
\STATE \hspace{0.5cm} \textbf{if} constraint for $\psi_b$ is violated\\
\STATE \hspace{1.3cm} set $\underline{\psi_b} = \psi_b$ \\
\STATE \hspace{0.5cm} \textbf{else} $\overline{\psi_b} = \psi_b$\\
\STATE \hspace{0.2cm} \textbf{for:} $ k_b = 1:K_b$  \\
\STATE \hspace{0.5cm} Compute $\hat{\bmG}_{k_b,b}^{UL}, \hat{\bmG}_{{k_b},b}^{DL},\hat{\bmG}_{k_b,\overline{b}}^{UL},\hat{\bmG}_{k_b,\overline{b}}^{DL} $ from Table \ref{table_gradients}.
\STATE \hspace{0.5cm}
Optimize $\bmU_{k_b}$ with \eqref{digital_solution_UL} and normalize its columns \\
\STATE \hspace{0.9cm} \textbf{Repeat until convergence}\\
\STATE \hspace{1.1cm} set $\lambda_{k_b} = (\underline{\lambda_{k_b}} + \overline{\lambda_{k_b}})/2$  \\ 
\STATE \hspace{1.1cm} Compute $\mathbf{P}_{k_b}$ with \eqref{optimal_pow}, do SVD, set $\mathbf{P}_{k_b} = \bmD_{P_{k_b}}^{svd}$ \\ \hspace{1.1cm} Set  $\bmT_{_b} = \bmU_{k_b} \mathbf{P}_{k_b} \bmU_{k_b}^H$
\STATE \hspace{1.1cm} \textbf{if} constraint for $\lambda_{k_b}$ is violated\\
\STATE \hspace{1.3cm} set $\underline{\lambda_{k_b}}= \lambda_{k_b}$  \\
\STATE \hspace{1.1cm} \textbf{else} $\overline{\lambda_{k_b}} = \lambda_{k_b}$\\
\STATE \hspace{0.2cm} Next $k_b$
\STATE \hspace{0.001cm} Next $b$
\end{algorithmic}
\textbf{Quantize}  $\bmW_b$ and $\bmF_b,$ with $\mathbb{Q}_b(\cdot), \forall b$\
\label{algo1}  
\end{algorithm}

\subsection{Convergence of C-HYBF}

The convergence of Algorithm \ref{alg_1} can be proved by using the 
minorization theory \cite{stoica2004cyclic}, alternating or cyclic optimization \cite{stoica2004cyclic}, Lagrange dual function \cite{boyd2004convex}, saddle-point interpretation \cite{boyd2004convex} and KKT conditions \cite{boyd2004convex}. For the WSR cost function \eqref{WSR}, we construct its minorizer, which is a touching lower bound for \eqref{WSR}, hence we can write

\begin{equation}
\begin{aligned}
\mbox{WSR} & \geq   \underline{\mbox{WSR}}  = \underline{\mbox{WR}}_{k_b,b}^{UL} + \underline{\mbox{WSR}}_{\overline{k}_b,b}^{UL} + \underline{\mbox{WR}}_{j_b,b}^{DL} +\underline{\mbox{WSR}}_{\overline{j}_b,b}^{DL} \\&  
+\underline{\mbox{WSR}}_{\overline{b}}^{DL} +\underline{\mbox{WSR}}_{\overline{b}}^{UL}.
\end{aligned}
\label{eqWSR2}
\end{equation}

The minorized WSR, which is concave in $\mathbf{T}_{k_b}$ and $\mathbf{Q}_{j_b}$,  has the same gradient of the original WSR maximization problem \eqref{WSR}, hence the KKT conditions are not affected.
Reparameterizing $\mathbf{T}_{k_b}$ or $\mathbf{Q}_{j_b}$ in terms of $\bmW_b, \mathbf{V}_{j_b}, \forall j_b \in \mathcal{D}_b$, or $\mathbf{U}_{k_b}, \forall k_b \in \mathcal{U}_b$, respectively, augmenting the minorized WSR cost function with the Lagrange multipliers and power constraints leads to \eqref{Largrangian}. By incorporating further the power matrices we get to $\mathcal{L}(\mathbf{\Lambda},\mathbf{\Psi},\mathbf{P}^{DL},\mathbf{P}^{UL})$. Every alternating update of the $\mathcal{L}$ for the variables $\bmW_b, \bmF_b, \forall b \in \mathcal{B}, \mathbf{V}_{j_b}, \forall {j_b} \in \mathcal{D}_b, \mathbf{U}_{k_b}, \forall k_b \in \mathcal{U}_b, \bmP_{k_b}, \bmP_{j_b}, \lambda_{k_b}$ and $\psi_b$, leads to a monotonic increase of the WSR, which assures
convergence. For the KKT conditions, at the convergence point, the gradients of $\mathcal{L}$ for $\mathbf{V}_{j_b},\bmW_b, \mathbf{U}_{k_b}$ or $\bmP_{k_b}, \bmP_{j_b}$ correspond to the gradients of the Lagrangian of the original
problem \eqref{WSR}, and hence the sub-optimal solution for the minorized WSR matches the sub-optimal solution of the original problem. For the fixed analog and digital beamformers, $\mathcal{L}$ is concave in powers, hence we have strong duality for the saddle point, i.e.,
\begin{equation}
    \max_{\bmP^{DL},\bmP^{UL}} \min_{\mathbf{\Lambda},\mathbf{\Psi}} \quad \mathcal{L}(\mathbf{\Lambda},\mathbf{\Psi},\bmP^{UL},\bmP^{DL}). 
\end{equation}
Let $\bmX^*$ and $x^*$
denote the optimal solution for matrix $\bmX$ or scalar $x$ at the convergence, respectively. As each iteration leads to a monotonic increase in the WSR and the power are updated by satisfying the sum-power constraint, at the convergence point, the solution of the optimization problem 
\begin{equation}
    \min_{\mathbf{\Lambda},\mathbf{\Psi}} \quad \mathcal{L}(\bmV_{j_b}^*,{\bmW_b}^{*},{\bmF_b}^{*},\bmU_b^*,{\bmP^{DL}}^*,{\bmP^{UL}}^*,\mathbf{\Lambda},\mathbf{\Psi})
\end{equation}
satisfies the KKT conditions for the powers in $\bmP^{DL}$ and $\bmP^{UL}$ and the complementary slackness conditions
\begin{subequations}
\begin{equation}
\psi_b^* \, (p_b - \sum_{j_b\in \mathcal{D}_b} \mbox{Tr}( {\bmW_b}^{*}\mathbf{V}_{j_b}^{*} \bmP_{j_b}^*\bmV_{j_b}^{*\,  H}{\bmW_b}^{*\,H})) = 0, 
\end{equation}\begin{equation}
\lambda_{k_b}^* \, (p_{k_b} -  \mbox{Tr}( \bmU_{k_b}^{*} \bmP_{k_b}^*\bmU_{k_b}^{*\,  H}))  = 0,
\end{equation}
\label{eqslackness}
\end{subequations}
with the individual factors in the products being non-negative.

\section{Parallel and Distributed Implementation} \label{Distributed_HYBF}
%As we discussed in Section \ref{Intro}, implementation of C-HYBF in a real-time large mmWave FD network is infeasible. 
%To overcome its drawbacks, we present the concept of P$\&$D-HYBF for mmWave based on cooperation among FD BSs. To develop the P$\&$D-HYBF scheme, we assume the following:
As previously mentioned in Section \ref{Intro}, the implementation of C-HYBF in a real-time, large-scale mmWave FD network is not feasible. In order to address these limitations, we introduce the concept of P$\&$D-HYBF for mmWave, which relies on cooperative interactions among neighbouring FD BSs. To establish the P$\&$D-HYBF scheme, we make the following assumptions:
\begin{enumerate}
   \item there exists a \emph{feedback link among the neighbouring FD BSs} and they cooperate by exchanging information about the digital beamformers, analog beamformers and analog combiners via the feedback link;
  \item \emph{local CSI} is accessible by each FD BS;
    \item each FD BS has \emph{multiple low-cost computational processors} dedicated for UL and DL;
    \item computations take place at the FD BSs in each cell in a \emph{synchronous} fashion, i.e., iteration $n$ at each BS takes place when iteration $n-1$ is completed by all the BSs.
\end{enumerate}  
Note that to satisfy 1), information can be broadcasted as well.
Recall that the MM optimization technique allowed us to write the Lagrangian of the original WSR problem \eqref{WSR} as 
 \begin{equation}\label{Largrangian}
\begin{aligned}
          \mathcal{L} = &\hspace{-1mm}\sum_{b \in \mathcal{B}}   \hspace{-1mm} \big[ \sum_{k_b \in \mathcal{U}_b} (w_{k_b} \mbox{ln det}( \bmI + \bmU_{k_b}^H \mathbf{\Sigma}_{k_b}^{1} \bmU_{k_b})
        - \mbox{Tr}(\bmU_{k_b}^H \mathbf{\Sigma}_{k_b}^{2} \bmU_{k_b}) \\& + \lambda_{k_b} p_{k_b}  )
            +\hspace{-1mm} \sum_{j_b \in \mathcal{D}_b} ( w_{j_b} \mbox{ln det}( \bmI +  \bmV_{j_b}^H {\bmW_b}^H \mathbf{\Sigma}_{j_b}^{1} {\bmW_b} \bmV_{j_b}) \\& -\mbox{Tr}(\bmV_{j_b}^H {\bmW_b}^H  \mathbf{\Sigma}_{j_b}^{2} {\bmW_b} \bmV_{j_b}) +  \psi_b p_b)   \big].
\end{aligned}
\end{equation}
%By analyzing its structure, it is to be noted that when optimizing one variable, assuming the remaining variables to be fixed, only the gradients appearing in $\mathbf{\Sigma}_{k_b}^{2}$ or $\mathbf{\Sigma}_{j_b}^{2}$ are required. Therefore, the gradients summarize complete information about all the remaining \emph{interfering} links in the network. From a practical point-of-view, the gradients for each link take into account the interference generated towards all the other links and hence limit greedy behaviour while optimizing its beamforming directions. However, $\eqref{Largrangian}$ is coupled among different links as the covariance matrices of other users directly appear in the gradients, which vary at the update of each beamformer.
Upon analyzing its structure, it is important to note that when optimizing a single variable, assuming the remaining variables are fixed, only the gradients present in $\mathbf{\Sigma}_{k_b}^{2}$ or $\mathbf{\Sigma}_{j_b}^{2}$ are required. Therefore, the gradients provide comprehensive information regarding the remaining \emph{interfering} links within the network. From a practical standpoint, the gradients for each link account for the interference generated towards all other links, thereby discouraging greedy behavior during the optimization of beamforming directions. However, it should be acknowledged that $\eqref{Largrangian}$ exhibits coupling among different links, as the covariance matrices of other users directly influence the gradients, which vary during the update of each beamformer.

To decouple \eqref{Largrangian} into local per-link independent optimization sub-problems, we assume that each FD BS has some memory to save information. Hereafter, overline will emphasize that the variables are only local and saved in the memory of each FD BS. We introduce the following local variables 
\begin{subequations}  \label{local_var_1} 
\begin{equation} 
\overline{\bmL}_{k_b}^{In} = \hat{\bmG}_{\overline{k}_b,b}^{UL}+ \hat{\bmG}_{k_b,b}^{DL},
  \; 
 \overline{\bmL}_{k_b}^{Out} = \hat{\bmG}_{k_b,\overline{b}}^{UL}+ \hat{\bmG}_{k_b,\overline{b}}^{DL}, \quad \forall k_b \in \mathcal{U}_b,
 \end{equation}
\begin{equation}  
\overline{\bmL}_{j_b}^{In}=  \hat{\bmG}_{j_b,b}^{UL} + \hat{\bmG}_{\overline{j}_b,b}^{DL},
 \; \; 
    \overline{\bmL}_{j_b}^{Out} = \hat{\bmG}_{j_b,\overline{b}}^{UL} + \hat{\bmG}_{j_b,\overline{b}}^{DL}, \quad \forall j_b \in \mathcal{D}_b,
\end{equation}
\end{subequations}
to be saved in the memory of each FD BS $b \in \mathcal{B}$. The local variables $\overline{\bmL}_{k_b}^{In}$ and $\overline{\bmL}_{k_b}^{Out}$ store information pertaining to the total interference generated within and outside the cell, respectively, by the beamformer of the UL user $k_b \in \mathcal{U}_b$. Similarly, the local variables $\overline{\bmL}_{j_b}^{In}$ and $\overline{\bmL}_{j_b}^{Out}$ capture information about the interference generated within the same cell and in neighbouring cells, respectively, by the FD BS $b \in \mathcal{B}$ while serving its DL user $j_b \in \mathcal{D}_b$. It is important to note that each FD BS can update the in-cell local variables $\overline{\bmL}_{k_b}^{In}$ and $\overline{\bmL}_{j_b}^{In}$ by itself, without requiring feedback. To update the out-cell variables $\overline{\bmL}_{k_b}^{Out}$ and $\overline{\bmL}_{j_b}^{Out}$, feedback from neighbouring BSs regarding their beamformers is necessary for computing \eqref{local_var_1}. Additionally, to store information regarding the interference-plus-noise covariance matrices at the RF chains and antenna level (for the analog combiner), the following local variables are defined
%The local variables $\overline{\bmL}_{k_b}^{In}$ and $\overline{\bmL}_{k_b}^{Out}$ save information about the overall interference generated inside and outside the cell by the beamformer of UL user $k_b \in \mathcal{U}_b$, respectively. Similarly, the local variables $\overline{\bmL}_{j_b}^{In}$ and $\overline{\bmL}_{j_b}^{Out}$ save information about the interference generated in the same cell and in the neighbouring cells by the FD BS $b \in \mathcal{B}$, respectively, while serving its DL user $j_b \in \mathcal{D}_b$. Note that each FD BS can update the in-cell local variables
%$\overline{\bmL}_{k_b}^{In}$ and $\overline{\bmL}_{j_b}^{In}$ by itself. Feedback from neighbouring BSs about their beamformers is required only to update the out-cell variables $\overline{\bmL}_{k_b}^{Out}$ and $\overline{\bmL}_{j_b}^{Out}$, such can \eqref{local_var} can be computed $\forall k_b, \forall j_b$. To save information about the interference-plus-noise covariance matrices at the RF chains and antenna level (for the analog combiner), we define the following local variables
\begin{subequations} \label{local_covariance}
\begin{equation} 
    \overline{\bmR}_{\overline{j}_b}^{-1} = \bmR_{\overline{j}_b}^{-1}, \quad \forall j_b, 
\end{equation} 
\begin{equation}
\overline{\bmR}_{\overline{k}_b}^{-1} = \bmR_{\overline{k}_b}^{-1}, \quad \overline{\bmR}_{k_b}^a = \bmR_{k_b}^a, \quad \overline{\bmR}_{\overline{k}_b}^a = \bmR_{\overline{k}_b}^a, \quad \forall k_b.
\end{equation}
\end{subequations}
For notational compactness, similar to \eqref{sigma_local_var}, we also define the following variables 
\begin{subequations} \label{local_var}
\begin{equation}
  \mathbf{Z}_{k_b}^{1} = \bmH_{k_b}^H \bmF_b \overline{\bmR}_{\overline{k}_b}^{-1} {\bmF_b}^H \bmH_{k_b}, \quad
       \mathbf{Z}_{k_b}^{2} =\overline{\bmL}_{k_b}^{In}+ \overline{\bmL}_{k_b}^{Out} + \lambda_{k_b} \bmI , \quad \forall k_b,
\end{equation}
\begin{equation}
  \mathbf{Z}_{j_b}^{1} = \bmH_{j_b}^H \overline{\bmR}_{\overline{j}_b}^{-1} \bmH_{j_b},  \quad   \quad \quad \;\;\;
     \mathbf{Z}_{j_b}^{2} = \overline{\bmL}_{j_b}^{In}+ \overline{\bmL}_{j_b}^{Out} + \psi_{b}\bmI, \quad \forall j_b,
\end{equation}\label{diagonal_matrices}
\end{subequations} 
which now depends only on the fixed local variables. By using the local variables, the Lagrangian \eqref{Largrangian} can be rewritten in terms of the information saved only in the local variables as  
\begin{equation}\label{Largrangian_decoupled}
\begin{aligned}
          \mathcal{L}  = &\sum_{b \in \mathcal{B}}    [\sum_{k_b \in \mathcal{U}_b} ( w_{k_b} \mbox{ln det}( \bmI + \bmU_{k_b}^H \mathbf{Z}_{k_b}^{1}\bmU_{k_b} )
            - \mbox{Tr}(\bmU_{k_b}^H \mathbf{Z}_{k_b}^{2}  \bmU_{k_b})   \\& +  \lambda_{k_b} p_{k_b} )
            +  \sum_{j_b \in \mathcal{D}_b} ( 
 w_{j_b} \mbox{ln det}( \bmI + \bmV_{j_b}^H {\bmW_b}^H \mathbf{Z}_{j_b}^{1} \bmW_b \bmV_{j_b} ) \\& -\mbox{Tr}(\bmV_{j_b}^H {\bmW_b}^H  \mathbf{Z}_{j_b}^{1} \bmW_b \bmV_{j_b})  +  \psi_b p_b) \big].
\end{aligned}
\end{equation} 
%In contrast to \eqref{Largrangian}, \eqref{Largrangian_decoupled} becomes fully decoupled per cell (not per link) as it is a function of the local variables, which are fixed for each FD BS and are updated solely when feedback from the neighbouring BSs is received. However, note that in UL and DL, the optimization of the analog combiner $\bmF_b$ in $\mathbf{Z}_{k_b}^{1}$ and analog beamformer $\bmW_b$, $\forall b \in  \mathcal{B}$, is still coupled as they are common among the UL and DL users in the same cell, respectively. Moreover, the analog beamformer $\bmW_b$ also affects the total transmit power of each FD BS, posing a serious challenge for enabling per-link independent optimization in DL from \eqref{Largrangian_decoupled}. Handling of the coupling constraints and P$\&$D optimization for HYBF from \eqref{Largrangian_decoupled} $\forall b \in \mathcal{B}$ is discussed in the following.
Unlike \eqref{Largrangian}, \eqref{Largrangian_decoupled} becomes fully decoupled at the cell level (not per link). This is because it relies on local variables that remain fixed for each FD BS and are updated only when feedback is received from neighbouring BSs. However, it should be noted that in both UL and DL, the optimization of the analog combiner $\bmF_b$ in $\mathbf{Z}_{k_b}^{1}$ and the analog beamformer $\bmW_b$ for all $b \in \mathcal{B}$ remains coupled, as they are shared among the UL and DL users within the same cell, respectively. Additionally, the analog beamformer $\bmW_b$ also impacts the total transmit power of each FD BS, presenting a significant challenge in achieving independent per-link optimization in the DL based on \eqref{Largrangian_decoupled}. The management of the coupling constraints and the P$\&$D optimization for HYBF, starting from \eqref{Largrangian_decoupled} $\forall b \in \mathcal{B}$ is discussed in the following.

%\begin{subequations}\label{tot_local}
% \begin{equation}\label{WSR_convex_local_var}
%\begin{aligned}
 %         \underset{\substack{\bmU,\bmV,\\\bmW_b,\bmF_b}}{\max} & \quad \hspace{-2mm} \sum_{b \in \mathcal{B}} \sum_{k_b \in \mathcal{U}_b} [w_{k_b} \mbox{ln det}( \bmI + \bmU_{k_b}^H \bmH_{k_b}^H {\bmF_b} \overline{\bmR}_{\overline{k}_b}^{-1} {\bmF_b}^H \\&\bmH_{k_b} \bmU_{k_b} )
           %- \mbox{Tr}(\bmU_{k_b}^H %(\overline{\bmL}_{k_b}^{Out} + \overline{\bmL}_{k_b}^{In})  \bmU_{k_b})]
         % \\& + \sum_{b \in \mathcal{B}} \sum_{j_b \in \mathcal{D}_b} w_{j_b} [\mbox{ln det}( \bmI + \bmV_{j_b}^H {\bmW_b}^H \bmH_{j_b}^H \overline{\bmR}_{\overline{j}_b}^{-1}\bmH_{j_b} \\& \bmW_b \bmV_{j_b} ) -\mbox{Tr}(\bmV_{j_b}^H {\bmW_b}^H ( \overline{\bmL}_{j_b}^{Out} + \overline{\bmL}_{j_b}^{In} ) \bmW_b \bmV_{j_b})].
%\end{aligned}
%\end{equation}
%\begin{equation}
 %   \mbox{s.t.} \hspace{3mm} \eqref{c1}-\eqref{c4}
%\end{equation}
%\end{subequations}

%Note that \eqref{WSR_convex_local_var} has the same structure of \eqref{WSR_convex}, but by replacing the gradients with the fixed local variables, the global WSR problem decouples into per-link independent optimization sub-problems. Optimization of the analog combiners and analog beamformers is still coupled as they are common to all the UL and DL users in the same cell, respectively. Also, optimization of the digital beamformers for the DL users in the same cell remains coupled as each BS has to satisfy the sum power constraint. Their decoupling and the solution of \eqref{tot_local} is discussed in the following.

\subsection{Per-Link Independent Sub-Problems in UL}
%In UL, each user has its own sum-power constraint but the analog combiner $\bmF_b$, appearing in $\textbf{Z}_{k_b}^1$, is common among all the UL users in the same cell. To decouple the optimization into per-link, we assume that FD BS $b \in \mathcal{B}$ updates $\bmF_b$ only after updating all the digital beamformers $\bmU_{k_b}, \forall k_b \in \mathcal{U}_b$. Given this assumption and fixed local variables, the UL WSR maximization problem for each FD BS reduces into three layers of sub-problems. Namely, at the bottom layer, FD BS $b \in \mathcal{B}$ has to solve independent sub-problems to update $\bmU_{k_b}$, whose optimization is fully decoupled and therefore can be executed in parallel $\forall k_b$. At the middle layer, FD BS $b \in \mathcal{B}$ has to independently update the stream power matrix $\bmP_{k_b}$ while searching the multiplier $\lambda_{k_b}$, $\forall k_b$, for its independent power constraint. Finally, at the top layer, once the two-layer UL sub-problems are solved, only one update of the common analog combiner is required. Fig. \ref{decomposition_UL} highlights the idea of the proposed per-link decomposition for the UL WSR for FD BS $b \in \mathcal{B}$ into three sub-layers, and the sub-problems at each layer must be solved from the bottom to the top.

In UL transmission, each user is subject to an individual sum-power constraint, while the analog combiner $\bmF_b$, present in $\textbf{Z}_{k_b}^1$, is shared among all UL users within the same cell. To achieve decoupled optimization on a per-link basis, we make the assumption that FD BS $b \in \mathcal{B}$ updates $\bmF_b$ only after updating all digital beamformers $\bmU_{k_b}$ for all $k_b \in \mathcal{U}b$. With this assumption and fixed local variables, the UL WSR maximization problem for each BS can be decomposed into into independent three layers of sub-problems. 

At the bottom layer, FD BS $b \in \mathcal{B}$ independently and in parallel solves sub-problems to update $\bmU_{k_b}, \forall k_b$. At the middle layer, FD BS $b \in \mathcal{B}$ independently updates the stream power matrix $\bmP_{k_b}$ while searching for the independent multiplier $\lambda_{k_b}, \forall k_b$ in parallel. Finally, at the top layer, after solving the two-layer UL sub-problems, only one update of the shared analog combiner is required. Figure \ref{decomposition_UL} visually illustrates the concept of per-link decomposition for UL WSR for FD BS $b \in \mathcal{B}$, comprising three sub-layers, with the sub-problems at each layer being solved sequentially from bottom to top, given the local variables have been recently updated based on the feedback.

Due to the per-link independent decomposition with fixed local variables, the Lagrangian for the UL user $k_b \in \mathcal{U}_b$ with independent sum-power constraint $p_{k_b}$, can be written as 
\begin{equation} \label{Lagrangian_uk_parallel}
    \begin{aligned}
   \mathcal{L}_{k_b}  = w_{k_b} \mbox{ln det}( \bmI + \bmU_{k_b}^H \mathbf{Z}_{k_b}^{1}\bmU_{k_b} )
            - \mbox{Tr}(\bmU_{k_b}^H \mathbf{Z}_{k_b}^{2}  \bmU_{k_b})  +  \lambda_{k_b} p_{k_b},
    \end{aligned}
\end{equation}
in which for the bottom layer the analog combiner $\bmF_b$ in $\mathbf{Z}_{k_b}^1$ and the powers are fixed. To optimize $\bmU_{k_b}$, a derivative of \eqref{Lagrangian_uk_parallel} can be taken, which leads to a similar KKT condition as
\eqref{kkt_uplink}, with $\mathbf{\Sigma}_{k_b}^{i}$ replaced with $\mathbf{Z}_{k_b}^{i}, \forall i$.
By following a similar proof of Appendix A, it can be easily shown that the WSR maximizing $\bmU_{k_b}$ for \eqref{Lagrangian_uk_parallel} can be computed as
\begin{equation} \label{UL_BF_parallel}
 \begin{aligned}
        \bmU_{k_b} = \bmD_{d_{k_b}}(\mathbf{Z}_{k_b}^{1}, \mathbf{Z}_{k_b}^{2} ).
 \end{aligned}
 \end{equation}
%Note that \eqref{UL_BF_parallel} can be computed in parallel by the multi-processor FD BS $b \in \mathcal{B}$, $\forall k_b \in \mathcal{U}_b$ as the common part $\bmF_b$ will be updated at the top layer and information about the interference generated towards other links is captured in the local variables, equals to the gradients from the previous feedback. 
Note that the computation of \eqref{UL_BF_parallel} can be parallelized across the multi-processor FD BS $b \in \mathcal{B}$, $ \forall k_b \in \mathcal{U}_b$. This parallelization is feasible due to the fact that the shared component $\bmF_b$ will be updated at the top layer, and the information pertaining to the interference generated towards other links is encapsulated within the local variables, which are fixed.

%At the middle layer, the power optimization remains decoupled as each UL user has its own sum-power constraint and the local variables are fixed.  To find the optimal $\bmP_{k_b}$ in parallel $\forall k_b$, we first consider the normalization of the columns of \eqref{UL_BF_parallel} to unit-norm $\forall k_b$, and the independent power allocation problem for P$\&$D-HYBF in UL can be formally stated as

In the intermediate layer, the power optimization remains decoupled due to the individual sum-power constraints of each UL user and the fixed local variables. In order to determine the optimal $\bmP_{k_b}$ simultaneously $ \forall k_b$, we begin by normalizing the columns of \eqref{UL_BF_parallel} to have unit norm $\forall k_b$. The independent power allocation problem for P$\&$D-HYBF in the UL can be formally defined as follows:
\begin{equation} \label{power_UL_PD}
    \underset{\mathbf{P}_{k_b}}{\max} \quad [w_{k_b} \mbox{ln det}( \bmI + \bmU_{k_b}^H \mathbf{Z}_{k_b}^{1} \bmU_{k_b}  \mathbf{P}_{k_b} ) - \mbox{Tr}(\bmU_{k_b}^H \mathbf{Z}_{k_b}^{2} \bmU_{k_b}  \mathbf{P}_{k_b} )].  
\end{equation}
Solving \eqref{power_UL_PD} independently $\forall k_b$ yields the following parallel power allocation scheme 
\begin{equation}\label{optimal_pow_UL_per_link}
     \begin{aligned}
     \bmP_{k_b} = ( w_{k_b}  (&\bmU_k^H \mathbf{Z}_{k_b}^{1} \bmU_{k_b} )^{-1} - (\bmU_{k_b}^H \mathbf{Z}_{k_b}^{2} \bmU_{k_b} )^{-1} )^{+}, 
 \end{aligned}
 \end{equation}
which can be computed while searching for the multiplier $\lambda_{k_b}$ associated with its independent sum-power constraint in parallel $\forall k_b$. It is worth noting that the power allocation problem consists of a purely concave component and a linear component that is simultaneously convex and concave. As a result, the overall problem is concave, and the solution provided by \eqref{optimal_pow_UL_per_link} yields an optimal power allocation scheme. The multiplier $\lambda_{k_b}, \forall k_b$, must satisfy the condition that \eqref{Lagrangian_uk_parallel} is finite and that $\lambda_{k_b}$ is strictly positive. This multiplier can be determined by solving the following problem in parallel:
\begin{equation} \label{lag_min_max_parallel}
\begin{aligned}
   \underset{\lambda_{k_b}}{\min}\,\,& \underset{\bmP_{k_b}}{\max}  \quad  \mathcal{L}_{k_b}(\lambda_{k_b},\bmP_{k_b}), \\
  &\quad \quad \mbox{s.t.} \quad \quad \lambda_{k_b} \succeq 0,
\end{aligned}
\end{equation}
 while independently allocating the powers at the middle layer $\forall k_b$. The dual function
 \begin{equation}
     \underset{\bmP_{k_b}}{\max}  \quad  \mathcal{L}_{k_b}(\lambda_{k_b},\bmP_{k_b}),
 \end{equation} 
is convex \cite{boyd2004convex} and can be solved with the Bisection method, as for the C-HYBF scheme.  If $\bmP_{k_b}$ becomes non-diagonal, its diagonal structure can be reestablished as $\bmP_{k_b} = \bmD_{\bmP_{k_b}}^{svd}$, where $\bmD_{\bmP_{k_b}}^{svd}$ is a diagonal matrix obtained from SVD of the non-diagonal $\bmP_{k_b}$.
 
At the top layer, one update of $\bmF_b$ is required $\forall b \in \mathcal{B}$. Note that simultaneous variation in parallel of the beamformers $\bmU_{k_b}$ and powers $\bmP_{k_b}$, $\forall k_b,$ at the bottom and middle layer vary the received covariance matrices. This information should be updated in the local variables $\overline{\bmR}_{k_b}^a$ and $\overline{\bmR_{\overline{k}_b}^{a}}$ at the antenna level, which $\bmF_b$ should combine. As each FD BS $b\in \mathcal{B}$ has complete information about the  optimized variables at the middle and bottom layers, it can use it to update first $\overline{\bmR}_{k_b}^a$ and $\overline{\bmR_{\overline{k}_b}^{a}}, \forall k_b \in \mathcal{U}_b$. As the WSR is fully concave with respect to the analog combiner $\bmF_b$, to optimize it we consider the following optimization problem
for the unconstrained case
 \begin{figure*}[t]
    \centering
 \begin{minipage}{0.49\textwidth}
  \centering
\includegraphics[width=6cm,height=4cm,draft=false]{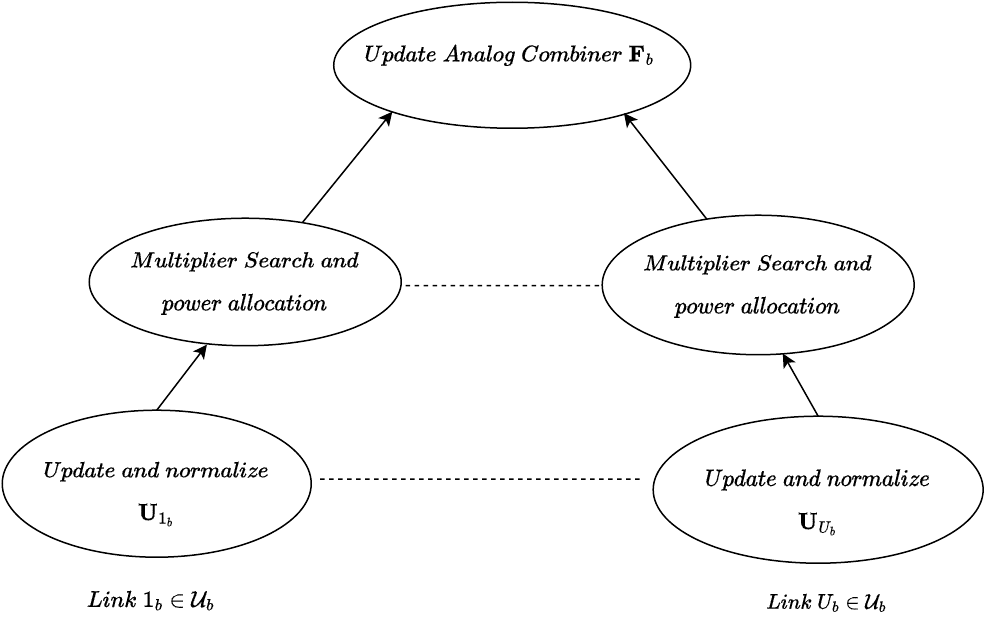}
    \caption{Decomposition of the UL WSR into three layers of sub-problems $\forall b \in \mathcal{B}$.}
    \label{decomposition_UL}
    \end{minipage} \hfill
      \begin{minipage}{0.49\textwidth}
   \centering
\includegraphics[width=6cm,height=4cm,draft=false]{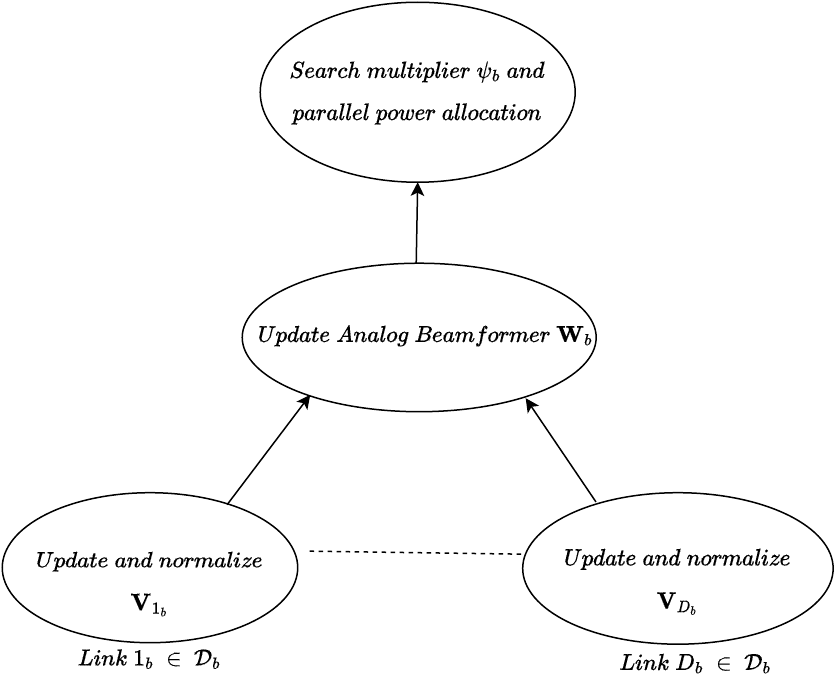}
    \caption{Decomposition of the DL WSR into three layers of sub-problems $\forall b \in \mathcal{B}$.}
    \label{decomposition_DL}
    \end{minipage}\hfill \vspace{-3mm}
\end{figure*} 
%The optimization of the unconstrained $\bmF_b$, given the recently updated local variables, can be formally stated as
\begin{equation} \label{analog_combiner_PandD}
\begin{aligned}
\underset{\substack{\bmF_b}}{\max}  \sum_{k_b \in \mathcal{U}_b}   w_{k_b}[ &  \mbox{ln det}({\bmF_b}^H \overline{\bmR}_{k_b}^a \bmF_b) -  \mbox{ln det}({\bmF_b}^H\overline{\bmR}_{\overline{k}_b}^a\bmF_b) ],
\end{aligned}
\end{equation}
where the local variables $\overline{\bmR}_{k_b}^a$ and $\overline{\bmR}_{\overline{k}_b}^a$ have been recently updated by using the information from the middle and bottom layers.
Problem \eqref{analog_combiner_PandD} is fully concave and solving it leads to the following optimal unconstrained analog combiner
\begin{equation} \label{analog_combiner_parallel}
             \bmF_b = \bmD_{N_b^{RF}}( \sum_{k_b \in \mathcal{U}_b} w_{k_b} \overline{\bmR}_{k_b}^a ,\sum_{k_b \in \mathcal{U}_b} w_{k_b} \overline{\bmR}_{\overline{k}_b}^a ).
\end{equation}
%To meet the unit-modulus and quantization constraints, we normalize its amplitudes with $\angle \cdot$ to unit-norm and quantize it as $\bmF_b(m,n) = \mathbb{Q}_b(\angle \bmF_b(m,n)) \in \mathcal{P}_b, \forall m,n$. Note that such operation results in optimality loss, which strictly depends on the phase resolution of the phase shifters and will be investigated with the simulations.
To satisfy the unit-modulus and quantization constraints, the amplitudes of $\bmF_b$ are normalized to have unit-norm and then quantized as $\bmF_b(m,n) = \mathbb{Q}_b(\angle \bmF_b(m,n)) \in \mathcal{P}_b$, for all $m$ and $n$. Such operation introduces a loss of optimality, which is directly influenced by the phase resolution of the phase shifters and will be examined through simulations.

\subsection{Per-Link Independent Sub-Problems in DL}
 
The decomposition of the DL WSR poses greater challenges due to the interdependence of the sum-power constraint among the DL users in set $\mathcal{D}_b, \forall b$. Additionally, the analog beamformer $\bmW_b$ is shared among DL users within the same cell and impacts the overall transmit power. To introduce per-link independent decomposition in DL, we assume that each FD BS $b \in \mathcal{B}$ first updates the digital beamformers for the DL users with unit-norm columns, while keeping the Lagrange multiplier $\psi_b$ and the analog beamformer $\bmW_b$ fixed. Furthermore, the powers are included afterwards, while searching the common multiplier $\psi_b$. Given this assumption, the DL WSR problem, for each FD BS $b \in \mathcal{B}$, decomposes into three layers of sub-problems.

At the bottom layer, each FD BS has to optimize the DL beamformers $\bmV_{j_b}$, $\forall j_b$, and normalize its columns to unit-norm, in parallel.
At the middle layer, one update of the analog beamformer $\bmW_b$ is required. Finally, at the top layer, we have to search for the Lagrange multiplier $\psi_b$ satisfying the coupled sum-power constraint and update the power matrices $\mathbf{P}_{j_b}$ for the DL users in parallel $\forall j_b$. Fig. \ref{decomposition_DL} shows the decomposition of the DL WSR into three layers of sub-problems, which must be solved from the bottom to the top.

We observe that the feasibility of per-link optimization arises from the exceptional attribute of MM, which permits power allocation at the highest layer. Alternative methodologies, like WMMSE, wherein powers and beamformers are jointly devised, would not facilitate the possibility of P$\&$D-HYBF.
  
For FD BS $b \in \mathcal{B}$, the Lagrangian for the DL WSR can be written as
  
\begin{equation} \label{lag_DL_vj}
\begin{aligned}
 \mathcal{L}_b^{DL} = & \sum_{j_b \in \mathcal{D}_b} ( w_{j_b} \mbox{ln det}( \bmI + \bmV_{j_b}^H {\bmW_b}^H \mathbf{Z}_{j_b}^{1} \bmW_b \bmV_{j_b} )  \\&  -\mbox{Tr}(\bmV_{j_b}^H  {\bmW_b}^H \mathbf{Z}_{j_b}^{2} \bmW_b \bmV_{j_b}  ) + \psi_b p_b ).
\end{aligned}
\end{equation}
%Note that in contrast to the UL case \eqref{Lagrangian_uk_parallel}, \eqref{lag_DL_vj} is coupled among the DL users in $\mathcal{D}_b$.
%In \eqref{lag_DL_vj}, for the bottom layer, as $\psi_b$, $\bmW_b$, $\mathbf{Z}_{j_b}^{i}$ are fixed, optimization of the digital beamformers with unit-norm columns
%remains decoupled. To optimize $\bmV_{j_b}$ a derivative can be taken which leads to a similar KKT condition as \eqref{kkt_downlink}, but now a function of the fixed $\mathbf{Z}_{j_b}^i$ which depend only on the fixed local variables only, and can be computed as

Note that in contrast to the UL case \eqref{Lagrangian_uk_parallel} which had per-link Lagrangian $\forall k_b \in \mathcal{U}_b$, \eqref{lag_DL_vj} exhibits coupling among the downlink users in $\mathcal{D}_b$.
In equation \eqref{lag_DL_vj}, for the lower layer, when $\psi_b$, $\bmW_b$, and $\mathbf{Z}{j_b}^{i}$ are held constant, the optimization of the digital beamformers with unit-norm columns remains independent. Given the local variables $\mathbf{Z}_{j_b}^i, \forall i$, digital beamformer $\bmV_{j_b}$ can be optimized as
 \begin{equation} \label{Vj_DL_parallel} 
     \begin{aligned}
     \bmV_{j_b} = \bmD_{d_{j_b}}( {\bmW_b}^H \mathbf{Z}_{j_b}^{1} \bmW_b , {\bmW_b}^H \mathbf{Z}_{j_b}^{2} {\bmW_b} ),
     \end{aligned}
 \end{equation}
The columns of which can be normalized to have a unit norm in parallel for all $j_b$, thereby enabling the inclusion of optimal power allocation at the top layer. Upon completing the parallel update of the digital beamformers $\bmV_{j_b} \forall j_b$, the middle layer necessitates the optimization of the analog combiner $\bmW_b$ for FD BS $b \in \mathcal{B}$. In this regard, each FD BS $b \in \mathcal{B}$ must independently tackle the following optimization problem at the middle layer for the unconstrained analog beamformer
 \begin{equation} \label{DL_problem_restated_analog_BF}
\begin{aligned}
\underset{\substack{\bmW_b}}{\max} \sum_{j_b \in \mathcal{D}_b} \quad   & ( w_{j_b} \mbox{ln det}( \bmI + \bmV_{j_b}^H {\bmW_b}^H \mathbf{Z}_{j_b}^{1}\bmW_b \bmV_{j_b} )    \\&  -\mbox{Tr}(\bmV_{j_b}^H  {\bmW_b}^H \mathbf{Z}_{j_b}^{2}\bmW_b \bmV_{j_b})).
\end{aligned}
\end{equation}
Note that each FD BS has complete information about the digital beamformers optimized at the bottom layer, which must be first used to update $\overline{\bmR}_{\overline{j}_b}^{-1}$ and $\overline{\bmL}_{j_b}^{In}$ appearing in $\mathbf{Z}_{j_b}^{1}$ and $\mathbf{Z}_{j_b}^{2}$ in \eqref{DL_problem_restated_analog_BF}, respectively, $\forall j_b$.
To optimize $\bmW_b$ a derivative of \eqref{DL_problem_restated_analog_BF} can be taken, and from the KKT condition it can be easily shown that $\bmW_b$ can be optimized as
 \begin{equation} \label{analog_BF_parallel}
     \begin{aligned}
      \mbox{vec}(\bmW_b)  =  \bmD_{1}( & \hspace{-1.5mm}\sum_{j_b \in \mathcal{D}_b} \hspace{-1.5mm}(\bmV_{j_b} \bmV_{j_b}^H (\bmI + \bmV_{j_b} \bmV_{j_b}^H {\bmW_b}^H \mathbf{Z}_{j_b}^{1} \bmW_b )^{-1})^T  \\& \otimes \mathbf{Z}_{j_b}^{1},\;  \sum_{j_b \in \mathcal{D}_b} (\bmV_{j_b} \bmV_{j_b}^H )^T  \otimes \mathbf{Z}_{j_b}^{2}).
    \end{aligned}  
 \end{equation}
 The analog combiner $\bmW_b$ optimized according to \eqref{analog_BF_parallel} is unconstrained and vectorized. Therefore, we do unvec$(\mbox{vec}(\bmW_b))$ to shape it into correct dimensions, normalize the amplitude with $\angle \cdot$ and quantize it such that $\bmW_b =\mathbb{Q}(\angle \bmW_b) \in \mathcal{P}_b$.
 
 For the top layer, the optimal stream power allocation can be included while searching the multiplier $\psi_b$ to satisfy the sum-power constraint $p_b$. Assuming the multiplier $\psi_b$ to be fixed, which is captured in $\mathbf{Z}_{j_b}^{2}$, the power optimization problem $\forall j_b \in \mathcal{D}_b$ can be stated as
 \begin{equation}\label{power_DL_parallel_problem}
\begin{aligned}
        \underset{\mathbf{P}_{j_b}}{\max}\quad & w_{j_b} \mbox{ln det}( \bmI + \bmV_{j_b}^H {\bmW_b}^H \mathbf{Z}_{j_b}^{1} \bmW_b \bmV_{j_b} \bmP_{j_b} )   \\& - \mbox{Tr}(\bmV_{j_b}^H  {\bmW_b}^H \mathbf{Z}_{j_b}^{2} \bmW_b \bmV_{j_b} \bmP_{j_b}).
\end{aligned}
\end{equation}
In \eqref{power_DL_parallel_problem}, the update of power matrix $\mathbf{P}_{j_b}, \forall j_b$, remains independent and the multiplier $\psi_b$ must be updated based on the sum of the transmit covariance matrices $\mbox{Tr}(\sum_{j_b} {\bmW_b}^H \bmV_{j_b}^H \bmP_{j_b} \bmV_{j_b}\bmW_b)$, once all the power matrices $\mathbf{P}_{j_b}$ are updated in parallel. Solving $\eqref{power_DL_parallel_problem}$ in parallel $\forall j_b$ leads to the following optimal power allocation scheme
 
 \begin{equation} \label{power_DL_parallel}
     \begin{aligned}
      \bmP_{j_b} = (& w_{j_b} (\bmV_{j_b}^H \bmW_b^H \mathbf{Z}_{j_b}^{1} \bmW_b \bmV_{j_b})^{-1}   - (\bmV_{j_b}^H \bmW_b^H \mathbf{Z}_{j_b}^{2} \bmW_b \bmV_{j_b})^{-1} )^{+}.
     \end{aligned}
 \end{equation}
As a final step, the multiplier $\psi_b$ can be searched with the Bisection method similar to \eqref{lag_min_max_parallel}, and while doing so, the optimal power allocation for the DL user in the set $\mathcal{D}_b$ can be computed in parallel according to \eqref{power_DL_parallel}. 

Adopting a cooperative methodology, each FD base BS engages in the exchange of information concerning the updated digital beamformers, analog beamformer, and combiner with neighboring BSs following each synchronized iteration. This collaborative information sharing enables every FD BS to gain insight into the generated interference and facilitates the appropriate adjustment of its beamformers throughout the independent optimization process. A comprehensive procedure outlining the execution of the cooperative P\&D-HYBF is presented in Algorithm 2.
 
\begin{algorithm} \footnotesize 
\caption{Parallel and Distributed Hybrid Beamforming}\label{alg_2}
\textbf{Given:} The rate weights, CSI and multiple processors in UL and DL $\forall b$.\\
\textbf{Initialize:}\;$\bmW_b, \bmV_{j_b}, \bmU_{k_b}, \forall j_b \in \mathcal{D}_b,\forall k_b \in \mathcal{U}_b$ in each cell.\\
\textbf{Repeat until convergence}
\begin{algorithmic}
%\STATE \hspace{0.001cm} \textbf{Cooperation stage} ($\forall b \in \mathcal{B})$
\STATE \hspace{0.5cm} $\forall b \in \mathcal{B}$, share $\bmW_b,\bmF_b$ and $\bmU_{k_b},\bmV_{j_b}, \forall k_b, \forall j_b$ with the neighbouring FD BSs.\\
\STATE \hspace{0.001cm} \textbf{In parallel} $\forall b$ ($\forall k_b \in \mathcal{U}_b, \forall j_b \in \mathcal{D}_b$)
\STATE \hspace{0.5cm} Update $\overline{\bmL}_{j_b}^{In},\overline{\bmL}_{k_b}^{In}$ and $\overline{\bmL}_{j_b}^{Out}$ and $\overline{\bmL}_{k_b}^{Out}$ in \eqref{local_var} , from the memory and based on the feedback, respectively.
\STATE \hspace{0.5cm} Update $\overline{\bmR}_{\overline{k}_b}^{-1}$ and $\overline{\bmR_{\overline{j}_b}}$ in \eqref{local_covariance} from the memory.
\STATE \hspace{0.001cm} \textbf{Solve in parallel} $\forall b \in \mathcal{B}$
\STATE \hspace{0.5cm} \textbf{Parallel DL for FD BS $b$}\\ 
\STATE \hspace{0.8cm} \textbf{Set: } $\underline{\psi_b}=0,\overline{\psi_b}= \psi_b^{max}$. \\
\STATE \hspace{0.8cm} \textbf{Bottom layer}: Compute $\bmV_{j_b},\forall j_b$ in parallel with \eqref{Vj_DL_parallel} and normalize it to unit-norm columns \\
\STATE \hspace{0.8cm} Update $\overline{\bmR}_{\overline{j}_b}^{-1}$ and $\overline{\bmL}_{j_b}^{In}, \forall j_b$ from the memory\\
\STATE \hspace{0.8cm} \textbf{Middle layer:} Compute $\bmW_b$ with \eqref{analog_BF_parallel}, do unvec and get $\angle \bmW_b$\\
\STATE \hspace{0.8cm} \textbf{Top layer:} \textbf{Repeat until convergence}\\
\STATE \hspace{1.3cm} set $\psi_b = (\underline{\psi_b} + \overline{\psi_b})/2$  \\
\STATE \hspace{1.3cm} \textbf{In parallel} $\forall j_b$
\STATE \hspace{1.7cm} Compute $\bmP_{j_b}$ \eqref{power_DL_parallel}, do SVD, set $\bmP_{j_b}=\bmD_{\bmP_{j_b}}$\\
\STATE \hspace{1.7cm} Set $\bmQ_{j_b} = \bmW_b \bmV_{j_b} \bmP_{j_b} \bmV_{j_b}^H \bmW_b^H $
\STATE \hspace{1.3cm} \textbf{if} constraint for $\psi_b$ is violated\\
\STATE \hspace{1.7cm} set $\underline{\psi_b} = \psi_b$ , \\
\STATE \hspace{1.3cm} \textbf{else} $\overline{\psi_b} = \psi_b$\\
\STATE \hspace{0.5cm} \textbf{Parallel UL for FD BS $b$}\\
\STATE \hspace{1.2cm} \textbf{Set: }  $\underline{\lambda_{k_b}}=0,\overline{\lambda_{k_b}}= \lambda_{k_b}^{max},\; \forall k_b$.
\STATE \hspace{1.2cm}  \textbf{Bottom layer:} Compute $\bmU_{k_b}, \forall k_b,$ \eqref{UL_BF_parallel} in parallel\\
\STATE \hspace{1.2cm} \quad  \quad  \quad  \quad  \quad   \quad  \; Normalize $\bmU_{k_b}, \forall k_b$.
\STATE \hspace{1.2cm} \textbf{Middle layer:} \textbf{Repeat until convergence in parallel $\forall k_b$}\\
\STATE \hspace{1.7cm} set $\lambda_{k_b} = (\underline{\lambda_{k_b}} + \overline{\lambda_{k_b}})/2$  \\ 
\STATE \hspace{1.7cm} Compute $\mathbf{P}_{k_b}$ with \eqref{optimal_pow_UL_per_link}, do SVD and set $\mathbf{P}_{k_b} = \bmD_{P_{k_b}}$\\ 
\STATE \hspace{1.7cm} Set $\bmT_{_b} = \bmU_{k_b} \mathbf{P}_{k_b} \bmU_{k_b}^H$\\
\STATE \hspace{1.7cm} \textbf{if} constraint for $\lambda_{k_b}$ is violated\\
\STATE \hspace{2cm} set $\underline{\lambda_{k_b}}= \lambda_{k_b}$ \\
\STATE \hspace{1.7cm} \textbf{else} $\overline{\lambda_{k_b}} = \lambda_{k_b}$ \\
\STATE \hspace{1.2cm} Update $\overline{\bmR}_{k_b}^a,\overline{\bmR}_{\overline{k}_b}^a$, $\forall k_b$.
\STATE \hspace{1.2cm} \textbf{Top layer:} Compute $\bmF_b$ with \eqref{analog_combiner_parallel} and get $\angle \bmF_b$. \\
\STATE \hspace{0.1cm} \textbf{Repeat}
\end{algorithmic}  
\textbf{Quantize} $\bmW_b$ and $\bmF_b,$ with $\mathbb{Q}_b(\cdot), \forall b$.
\label{algo2}  
\end{algorithm}

\emph{Remark 2:} Remark that the required information exchange is minimal as each FD BS only needs to share one analog beamformer and analog combiner, which have sizes corresponding to the number of antennas at the BS, along with several small-sized digital beamformers (sized according to the number of RF chains for DL or the number of UL user antennas, which is very small). Additionally, P$\&$D-HYBF relies on local interfering channels (Assumption 2), which are incorporated in the local variables, resulting in a significant reduction in communication overhead compared to C-HYBF. Furthermore, as computations take place locally at each FD BS, it allows faster adaptation of the beamformers to highly dynamically variable local traffic demands.
\vspace{-3mm}
\subsection{On the Convergence of P$\&$D-HYBF} 
The convergence of P$\&$D-HYBF can be shown by following a similar approach for C-HYBF and therefore we consider omitting it. However, there are some notable differences in the information captured by the local variables used in P$\&$D-HYBF compared to the gradients used in C-HYBF. This results in different KKT conditions and therefore both algorithms converge to a different local optimum during the alternating optimization process, resulting in different achievable WSR at each iteration. It will be shown later that the difference between the achievable WSR by C-HYBF and P$\&$D-HYBF at the convergence is negligible.

\vspace{-2mm}
\subsection{Computational Complexity Analysis} \label{analisi_complessita}
In this section, we present the per-iteration computational complexity for the C-HYBF and P$\&$D-HYBF schemes.
For such purpose, an equal number of users in DL and UL in each cell, i.e., $D_b=D$ and $U_b =U$, $\forall b \in \mathcal{B}$ with $|\mathcal{B}| = B$, is assumed. Moreover, the number of antennas in each cell for the FD BSs, UL and DL users is also assumed to be the same in each cell.  

Let us consider only the computational complexity of C-HYBF denoted with $\mathcal{C}_{C}$, by ignoring its huge communication overhead to transfer complete CSI every channel coherence time to the CN and communicating back the optimized variables, given by

%It is one iteration consists in updating $BD$ and $BU$ digital beamformers for the DL and UL users, respectively, and $B$ analog beamformers and $B$ analog combiners by the CN. Optimization of these many
%variables result in the following computational complexity
\begin{equation} \label{compl_C_HYBF}
        \begin{aligned}
            \mathcal{C}_{C} = & \mathcal{O}(B^2 U^2  {N_b^{RF}}^3 + B^2 U D N_{j_b}^3 + B^2 D^2 N_{j_b}^3  + B^2 D U {N_b^{RF}}^3 \\& + B {M_b^{RF}}^2 M_b^2 + B {N_b^{RF}} N_b^2   + B D d_{j_b} {M_b^{RF}}^2 + B D d_{k_b} N_{k_b}^2).
        \end{aligned}
    \end{equation}

For P$\&$D-HYBF, different computational processors have different computational burdens. Therefore we consider comparing its worst-case complexity in each cell, denoted as $\mathcal{C}_{P\&D}^{DL}$ and $\mathcal{C}_{P\&D}^{UL}$ for DL and UL, respectively. Namely, let $Z_{DL}$ and $Z_{UL}$ denote the number of processors dedicated for DL and UL by each FD BS, respectively. It is necessary to distinguish two cases: 
1) $Z_{DL} =D$, $Z_{UL}=U$, and 2) $Z_{DL} < D$, $Z_{UL} < U, \forall b$. The first case considers fully parallel implementation with the number of processors equals to the number of users. For such a case, the worst-case complexity in UL and DL is given for the processors which make one update of the digital beamformer in UL and the analog combiner, and one update of the digital beamformer in DL and the analog beamformer, respectively, in each cell. For the fully parallel implementation (case 1), the complexity results to be
 
\begin{subequations}
    \begin{equation}
    \begin{aligned}
         \mathcal{C}_{P\&D}^{DL} = \mathcal{O}(B D N_{j_b}^3 + B U {N_b^{RF}}^3 + d_{j_b} {M_b^{RF}}^2  + {M_b^{RF}}^2  M_b^2),
    \end{aligned}
    \end{equation}
    \begin{equation}
         \mathcal{C}_{P\&D}^{UL} = \mathcal{O}( B U {N_b^{RF}}^3  + B D N_{j_b}^3 + d_{k_b} N_{k_b}^2 + {N_b^{RF}} N_b^2)
    \end{equation}
\end{subequations} 

The second case considers that the number of processors dedicated is less than the number of users. In such a case, each processor may have to update $K$ and $N$ digital beamformers in DL and UL, respectively, before updating the analog part. In such a case, the worst-case complexity in DL and UL is given for the processors which update $K$ digital beamformers and the analog beamformer, and $N$ digital beamformers and the analog combiner, respectively, in each cell. 
In such a case, the worst-case complexity in UL and DL for P$\&$D-HYBF is

\begin{subequations}
    \begin{equation}
        \begin{aligned}
          \mathcal{C}_{P\&D}^{DL} = \mathcal{O}(K B D N_{j_b}^3 + K B U {N_b^{RF}}^3 + K d_{j_b} {M_b^{RF}}^2  + {M_b^{RF}}^2  M_b^2)
        \end{aligned}
    \end{equation}
    \begin{equation}
       \mathcal{C}_{P\&D}^{UL}= \mathcal{O}(N B U {N_b^{RF}}^3  + N B D N_{j_b}^3 + N d_{k_b} N_{k_b}^2 + {N_b^{RF}} N_b^2)
    \end{equation}
\end{subequations}
%We can see that the worst-case complexity of P$\&$D-HYBF is significantly less compared to C-HYBF, scaling only linearly as a function of network size $B$ and density, i.e., $D$ and $U$. On the other hand, C-HYBF not only requires significant computational complexity, but it also scales quadratically with the network size and density. Note that C-HYBF requires significant additional overhead to transfer full CSI to the CN, which is ignored in the computational. Furthermore, higher computational complexity also leads to higher power consumption at the CN for processing \cite{schneider2012could}. On the other hand, P$\&$D-HYBF can enable the deployment of low-cost multi-processor FD BSs, and such processors will require minimal power consumption for solving the per-link independent sub-problems. 

It is noteworthy that the worst-case complexity of P$\&$D-HYBF is significantly lower than that of C-HYBF. It scales linearly with the network size ($B$) and density ($D$ and $U$), while C-HYBF has a quadratic scaling. %Additionally, C-HYBF requires substantial overhead for transferring complete CSI to the CN, which is not accounted for in the computational complexity.  Consequently, P$\&$D-HYBF enables the deployment of cost-effective multi-processor FD BSs for solving local per-link independent low-complexity sub-problems.

 \section{Simulation Results} \label{simulazioni}
This section presents simulation results to evaluate the performance of the proposed C-HYBF and P$\&$D-HYBF schemes. For comparison, we consider the following benchmark schemes:
1) A centralized \emph{Fully Digital FD} scheme with the LDR noise. 2) A centralized \emph{Fully Digital HD} scheme with LDR noise, serving the UL and DL users by separating the resources in times.

To compare the performance with a fully digital HD system, we define the additional gain in terms of percentage for an FD system over an HD system as
\begin{equation}
    Gain = \frac{\mbox{WSR}_{FD} - \mbox{WSR}_{HD}}{\mbox{WSR}_{HD}} \times 100 [\%],
\end{equation}
where $\mbox{WSR}_{FD}$ and $\mbox{WSR}_{HD}$ are the network WSR for the FD and HD systems, respectively. We assume the same SNR level for all the FD BSs, defined as $\mbox{SNR} = p_b/\sigma_b^2$,
with transmit power $p_b$ and thermal noise variance $\sigma_b^2$.
We assume that the UL users and FD BSs transmit the same amount of power, i.e., $p_{k_b}=p_b, \forall k_b$. The thermal noise level for DL users is set as $\sigma_{j_b}^2 = \sigma_b^2, \forall j_b$. The total transmit power is normalized to $1$, and we choose the thermal noise variance to meet the desired SNR.
We simulate a multicell network of $B=2$ cells, with each FD BS serving one DL and one UL user.
P$\&$D-HYBF is evaluated on a computer consisting of $4$ computational processors, equal to the number of users in the network, i.e., fully parallel implementation, by using the parfor function in \MATLAB. The FD BSs are assumed to have $M_b =100$ transmit and $N_b = 60$ receive antennas, and the number of RF chains in transmission and reception is chosen as $M_b^{RF}\hspace{-2mm}=N_b^{RF}=32,16,10$ or $8$. The phase shifters are assumed to be quantized with a uniform quantizer $\mathbb{Q}_b(\cdot)$ of $10$ or $4$ bits. The DL and UL users are assumed to have  $N_{j_b}=N_{k_b}=5$ antennas and are served with $d_{j_b}=d_{k_b}=2$ data streams. The number of paths for each device is chosen to be $N_{k_b}^p=3$, and the AOA $\theta_{k_b}^n$ and AOD $\phi_{k_b}^n$ are assumed to be uniformly distributed in the interval $\mathcal{U}\sim [-30^\circ,30^\circ], \forall j_b, k_b$. We assume uniform-linear arrays (ULAs) for the FD BSs and users. For the FD BSs, the transmit and the receive array are assumed to be separated with distance $D_b=20~$cm with a relative angle $\Theta_b = 90^\circ$ and $r_{m,n}$ in \eqref{SI_Channel} is set given $D_b$ and $\Theta_b$. The Rician factor is chosen to be $\kappa_b=1$ and the rate weights are set to be $w_{k_b}=w_{j_b}=1$. Digital beamformers are initialized as the dominant eigenvectors of the channel covariance matrices of each user. The analog beamformers and combiners are initialized as the dominant eigenvectors of the sum of the channel covariance matrices across all the DL and UL users, respectively. Results reported herein are averaged over $100$ channel realizations. %Note that as we are assuming perfect CSI, the SI with HYBF can be cancelled only up to the LDR noise floor, which reflects the residual SI power.

 %\begin{table}
%\centering
  %  \caption{Simulation parameters to simulate a multicell mMIMO mmWave FD system.}
  %  \resizebox{9cm}{!}{%
  %  \begin{tabular}{|p{40mm}|p{40mm}|p{40mm}|}
  %   \hline
   %    \multicolumn{3}{|c|}{ Parameters } \\
   %    \hline
   %    Cells &\mbox{$B$}  & 2 \\ \hline
   %   UL and DL users &\mbox{$ U_b,D_b$}  & 1, $\forall b$ \\ \hline
    %  Data streams   &\mbox{$d_{j_b}$},\mbox{$d_{k_b}$} & 2, $\forall b$   \\ \hline
   %    BSs antennas  &\mbox{$M_{b}, N_b$} &  100, 60  \\ \hline
   %    Clusters and Paths  &\mbox{$N_{c,b}$},\mbox{$N_{p,b}$} & 3,3  \\ \hline
      % RF chains (BSs)  &\mbox{$ M_{b}^{RF} = N_{b}^{RF}$} & 10,12,16,32  \\  \hline
 %    Rx RF chains  &\mbox{$ N_r$} & 10,12,16,32  \\ \hline
   %    User antennas &\mbox{$M_{k_b} = N_{j_b}$}  &  5  \\  \hline
   %    DL user antennas  &\mbox{$N_j$} &  5  \\\hline
    %  Rician Factor   &\mbox{$\kappa_b$} & 1 \\ \hline
     % Array response   &$\mathbf{a}_{r}(\cdot)$  $\mathbf{a}_{t}(\cdot)$ &  ULA,ULA,ULA \\\hline
      % AoA,AoD  & \mbox{$\phi_{k_b}^n$},\mbox{$\phi_{k_b}^n$}  &  $\mathcal{U}$\mbox{$\sim [-30^{\circ},30^{\circ}]$} \\ \hline
       % Rate weights &$w_k,w_j$ & 1\\ \hline
        %Sum Power &$ p_{k_b},p_b$ & 1,1 \\ \hline
      % Uniform Quantizer &$\mathbb{Q}_b(\cdot)$ &  $4$ or $10$ bits\\\hline
      % Relative Angle   &$\Theta_b$ & $90^{\circ}$\\\hline
    % Array separation   &  $D_b$ & 20 cm\\\hline
    %\end{tabular}}\label{table_parametri}
%\end{table} 
 
 \begin{figure}[t]
     \centering
\includegraphics[width=7cm,height=4.5cm,draft=false]{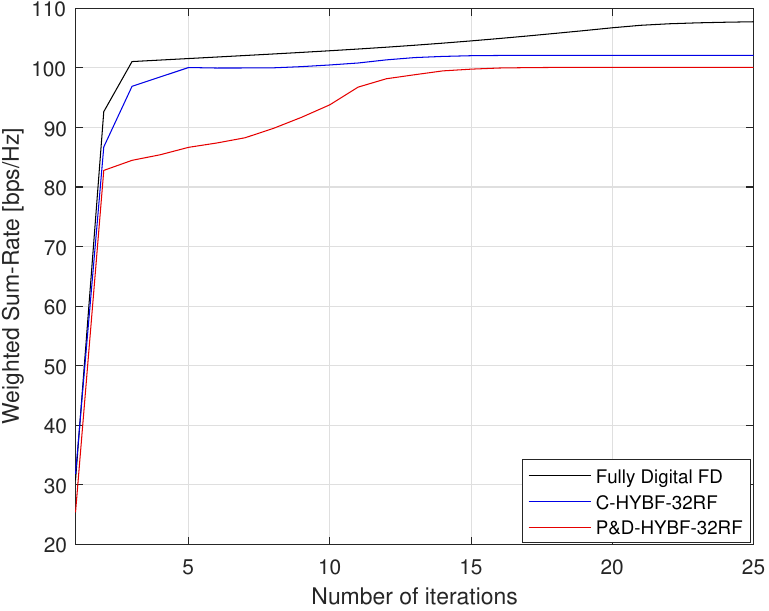}
     \caption{Typical convergence behaviour of the proposed schemes.}
     \label{convergence} \vspace{-7mm}
 \end{figure}

Figure \ref{convergence} showcases the convergence characteristics of the proposed schemes in comparison to the fully digital FD system. It is apparent that the proposed schemes achieve varying rates of improvement in WSR at each iteration and converge towards distinct local optima. However, the discrepancy in the attained WSR at convergence is negligible. Notably, P$\&$D-HYBF demonstrates rapid convergence within a few iterations, thereby requiring minimal communication overhead for sharing updated variable information.
%Figure \ref{tempo} showcases the execution time of C-HYBF and P$\&$D-HYBF, considering 32 RF chains. C-HYBF requires significant computational time as it updates variables iteratively through alternating optimization one after another. The transfer of complete CSI to the CN and communication of optimized variables to all cells further adds substantial time, disregarded in the simulations. In contrast, P$\&$D-HYBF enables parallel computation of local variables at each FD BS, wherein each BS optimize its local beamformers only for its associated users on multiple processors. The results demonstrate that P$\&$D-HYBF requires approximately 1/21 and 1/2.3 less time in UL and DL, respectively, compared to the average execution time of C-HYBF.
%The complexity of P$\&$D-HYBF in DL is dominated by the computation of a single large generalized eigenvalue problem (GDE) to update the vectorized analog beamformer, with a complexity of $\mathcal{O}({M_b^{RF}}^2 M_b^2)$. In UL, the complexity is primarily determined by the computation of the analog combiner, which scales as $\mathcal{O}({N_b^{RF}} N_b^2)$. Note that the execution time for solving one sub-problem for the bottom layers in UL and DL is negligible compared to the average execution time of C-HYBF.  Based on the complexity analysis above, we anticipate that the execution time of C-HYBF will grow quadratically with the number of users or cells. In contrast, P$\&$D-HYBF exhibits significantly smaller execution time, expected to increase linearly as the network expands.
 
Figure \ref{tempo} illustrates the execution time of C-HYBF and P$\&$D-HYBF, considering the utilization of 32 RF chains. C-HYBF requires substantial computational time as it iteratively updates variables through alternating optimization, one after another. The transfer of complete CSI to the CN and the communication of optimized variables to all cells further contribute to a significant time overhead, which is disregarded in the simulations. On the other hand, P$\&$D-HYBF allows for parallel computation of local variables at each FD BS, where each BS optimizes its local beamformers solely for its associated users using multiple low-cost processors. The results demonstrate that P$\&$D-HYBF requires approximately $1/21$ and $1/2.3$ less time in the UL and DL, respectively, compared to the average execution time of C-HYBF.

The complexity of P$\&$D-HYBF in the DL is dominated by the computation of a single large GDE to update the vectorized analog beamformer, with a complexity of $\mathcal{O}({M_b^{RF}}^2 M_b^2)$. In the UL, the complexity primarily stems from the computation of the analog combiner, which scales as $\mathcal{O}({N_b^{RF}} N_b^2)$. It is important to note that the execution time for solving one sub-problem for the bottom layers in the UL and DL is negligible compared to the average execution time of C-HYBF. Based on the aforementioned complexity analysis, we anticipate that the execution time of C-HYBF will grow quadratically with the number of users or cells which is infeasible for a large real-time FD network.
 
 \begin{figure*}[t]
    \centering
 \begin{minipage}{0.5\textwidth}
  \centering
\includegraphics[width=7cm,height=4.5cm,draft=false]{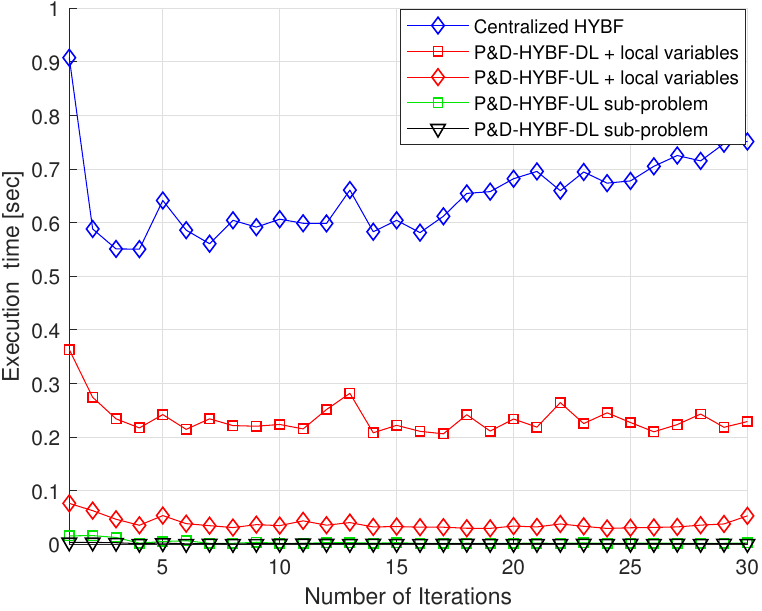}
     \caption{Execution time for the C-HYBF and the P$\&$D HYBF schemes with $32$-RF chains.}
     \label{tempo}
\end{minipage} \hfill
      \begin{minipage}{0.49\textwidth}
      \centering
\includegraphics[width=7cm,height=4.5cm,draft=false]{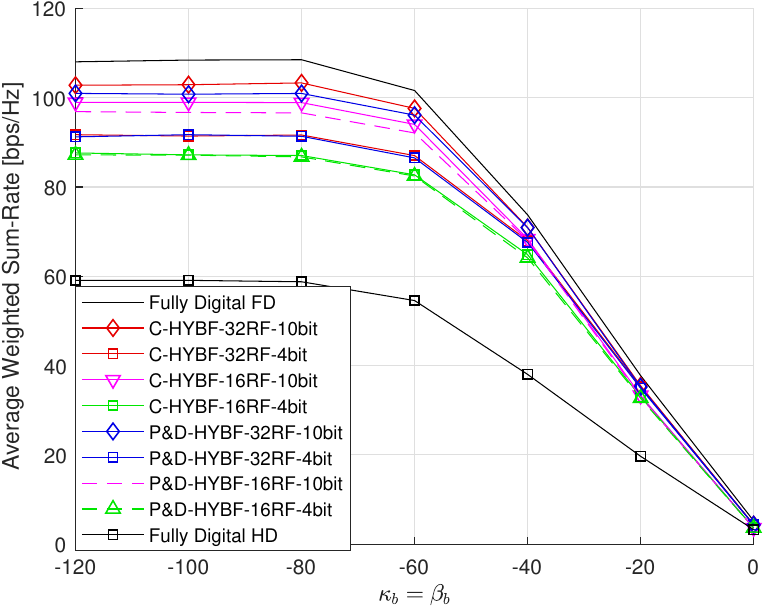}
    \caption{Average WSR as a function of the LDR noise with SNR$=20~$dB and a large number of RF chains. }
    \label{SNR_20dB_32RF}
    \end{minipage}  \vspace{-4mm}
\end{figure*}

\begin{figure}[t]
     \centering
\includegraphics[width=7cm,height=4.5cm,draft=false]{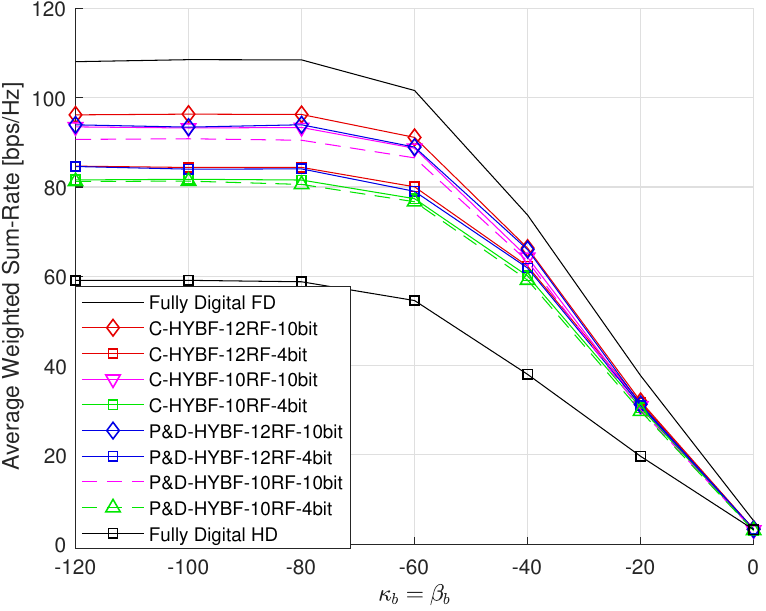}
    \caption{Average WSR as a function of the LDR noise with SNR$=20~$dB and a small number of RF chains.}
    \label{SNR_20dB_12RF} \vspace{-5mm}
\end{figure}

  \begin{figure*}[t]
    \centering
 \begin{minipage}{0.49\textwidth}
     \centering
\includegraphics[width=7cm,height=4.5cm,draft=false]{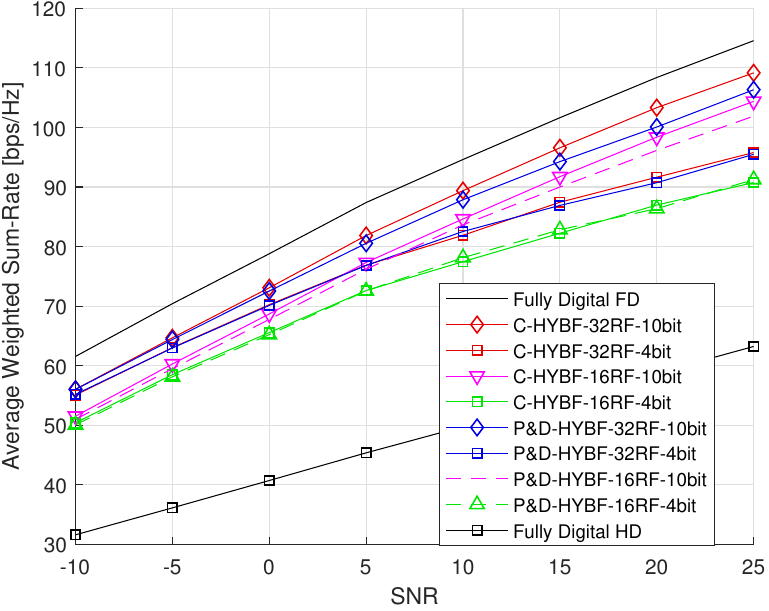}
    \caption{Average WSR as a function of the SNR with LDR noise $\kappa_{k_b}=-80~$dB and a large number of RF chains.}
    \label{SNR_LDR_08_32RF_16RF}
\end{minipage}  
      \begin{minipage}{0.49\textwidth}
     \centering
\includegraphics[width=7cm,height=4.5cm,draft=false]{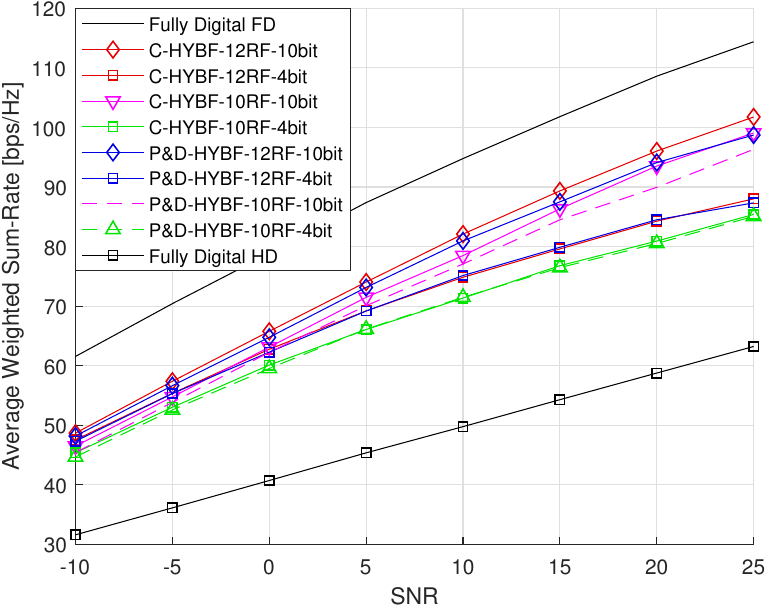}
    \caption{Average WSR as a function of the SNR with LDR noise $\kappa_{k_b}=-80~$dB and a small number of RF chains. }
    \label{SNR_LDR_08_12RF_10RF}
    \end{minipage}  
\end{figure*}

 Fig. \ref{SNR_20dB_32RF} shows the average WSR achieved with both schemes as a function of the LDR noise with $32$ or $16$ RF chains and $10$ or $4$ bits phase-resolution. It is visible that the P$\&$D-HYBF performs very close to the C-HYBF scheme with the same number of RF chains and phase resolution. Fully digital FD achieves $\sim 83 \%$ of additional gain than the fully digital HD for any LDR noise level. For a low LDR noise level $k_b < -80$ dB, C-HYBF and P$\&$D-HYBF with $32$ RF chains achieve $\sim74,55$ and $\sim 71,54\%$ additional gain and with $16$ RF chains they achieve $\sim 67\%, 48\%$ and $\sim 64\%, 47\%$ additional gain with $10,4$ bits phase-resolution, respectively. We can also see that when the LDR noise variance increases, the achievable WSR for both the FD and HD systems decreases considerably. Fig. \ref{SNR_20dB_12RF} shows the average WSR as a function of the LDR noise with only $12$ or $10$ RF chains and with $10$ or $4$ bits phase-resolution. For LDR noise $k_b \leq 80~$dB, C-HYBF and P$\&$D-HYBF with $12$ RF chains achieve $\sim 60, 43 \%$ and $\sim57, 43 \%$  additional gain and with $10$ RF chains they achieve additional gain of $\sim58, 38 \%$
 and $\sim 53, 37\%$ with $10,4$ bit phase-resolution, respectively. %From  \ref{SNR_20dB_32RF}-\ref{SNR_20dB_12RF} we can conclude the LDR noise generated from the small dynamic range hardware dictates the maximum achievable WSR with HYBF for both the schemes for mmWave FD, dictating the effective SINR. It is to be noted that P$\&$D-HYBF achieves similar performance as the C-HYBF scheme regardless of any LDR noise level.

 \begin{figure*}[t]
    \centering
 \begin{minipage}{0.49\textwidth}
     \centering
\includegraphics[width=7cm,height=4.5cm,draft=false]{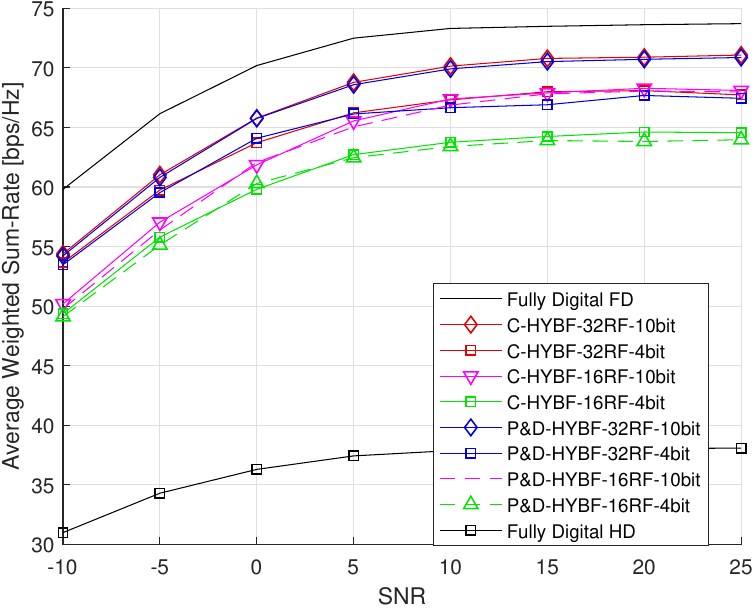}
    \caption{Average WSR as function of the SNR with LDR noise $\kappa_{k_b}=-40~$dB and a large number of RF chains. }
    \label{SNR_LDR_32_16RF_10RF}
\end{minipage}  
      \begin{minipage}{0.49\textwidth}
      \centering \includegraphics[width=7cm,height=4.5cm,draft=false]{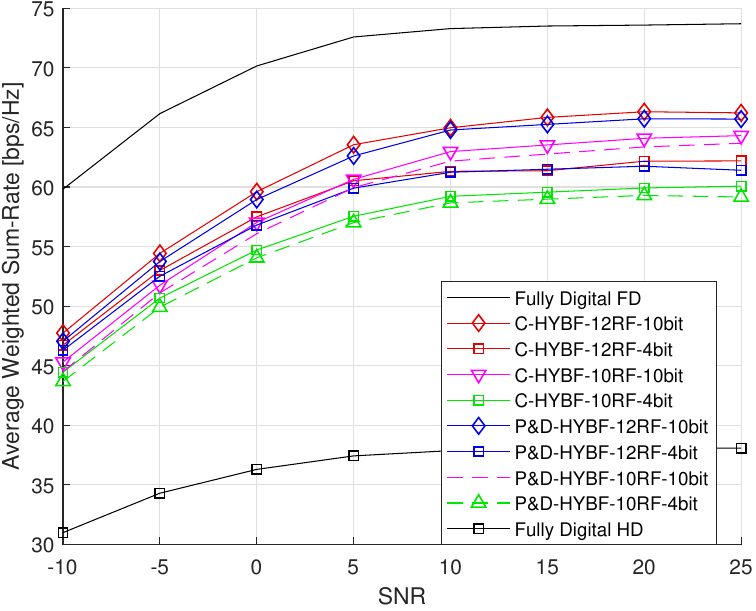}
    \caption{Average WSR as a function of the SNR with LDR noise $\kappa_{k_b}=-40~$dB and a small number of RF chains. }
    \label{SNR_LDR_04_12RF_10RF}
    \end{minipage}  
\end{figure*}

Fig. \ref{SNR_LDR_08_32RF_16RF} shows the average WSR as a function of the SNR with $32$ and $16~$RF chains and with $10$ or $4$ bit phase-resolution affected with LDR noise $k_b=-80~$dB, in comparison with the benchmark schemes. We can see that a fully digital FD system achieves $\sim 94\%$ and  $\sim 82\%$ additional gain at low and high SNR, respectively. With $32$ RF chains and $10$ bit phase-resolution, the C-HYBF scheme achieves $\sim 79 \%$ gain at all the SNR levels and the P$\&$D-HYBF achieves $\sim 77\%$ and $\sim 68\%$ gain at low and high SNR, respectively. As the phase resolution decreases to $4$-bits, we can see that the loss in WSR compared to the $10$-bit phase-resolution case is much more evident at high SNR. Still, with $16$ RF chains and $10$ or $4$ bit phase-resolution, both schemes significantly outperform the fully digital HD scheme for any SNR level. 

Fig. \ref{SNR_LDR_08_12RF_10RF} shows the average WSR as a function of the SNR with the same LDR noise level as in Fig. \ref{SNR_LDR_08_32RF_16RF}, i.e., $k_b=-80~$dB, but with $10$ or $12$ RF chains and $10$ or $4$ bit phase-resolution. The achieved average WSR presents a similar behaviour as in the case of a high number of RF chains, and it is visible that the proposed schemes significantly outperform the fully digital HD system also with a very low number of RF chains and phase-resolution. Moreover, P$\&$D-HYBF achieves similar performance as the C-HYBF scheme regardless of the phase resolution and number of RF chains.

Figure \ref{SNR_LDR_32_16RF_10RF} exhibits the average achieved WSR as a function of the SNR in the presence of LDR noise $k_b=-40$ dB. It is evident that when the LDR noise dominates, reducing the variance of thermal noise has a negligible impact on the effective signal-to-LDR-plus-thermal-noise ratio. Consequently, the dominance of LDR noise variance acts as a threshold, limiting the effective SLNR and thereby restricting the achievable WSR, which tends to saturate at SNR$=10$ dB. Notably, even with a high level of LDR noise, both the C-HYBF and P$\&$D-HYBF achieve similar performance gains. It can be also observed that both schemes achieve higher WSR when utilizing $16$ RF chains and $10$-bit phase resolution compared to the case of $32$ RF chains and $4$-bit phase resolution. Figure \ref{SNR_LDR_04_12RF_10RF} illustrates the average WSR as a function of the SNR, considering only $10$ or $12$ RF chains and $10$ or $4$-bit phase resolution. 

The aforementioned findings serve as compelling evidence for the superior performance of both proposed schemes compared to the fully digital HD system. Furthermore, an important conclusion can be drawn, namely that P$\&$D-HYBF exhibits comparable performance to C-HYBF without any degradation while offering numerous advantages in terms of scalability, computational complexity, deployability and communication overhead. As a result, P$\&$D-HYBF emerges as a prominent solution for future large-scale mmWave FD networks.

\section{Conclusion} \label{conclusioni}
This article presented two HYBF schemes for WSR maximization in multicell mmWave mMIMO FD systems. Firstly, a C-HYBF scheme based on alternating optimization is presented. However, C-HYBF requires massive communication overhead to exchange information between the FD network and the CN. Moreover, very high computational power is required to optimize numerous variables jointly. To overcome these drawbacks, a very low-complexity P$\&$D-HYBF design is proposed, which enables each FD BS to solve its local per-link independent sub-problems simultaneously on different computational processors, which drastically reduces the communication overhead. Its complexity scales only linearly as a function of the network size, making it highly scalable and enabling the deployment of low-cost computational processors at each BS.
Simulation results show that the proposed HYBF designs achieve similar average WSR and significantly outperform the centralized fully digital HD systems with only a few RF chains.

\appendices

\section{Proof of Theorem \ref{digital_BF_solution} }\label{AP.1}
%\label{hybrid_appendix}
%The proof of Theorem \ref{digital_BF_solution} follows similarly as in Appendix B \cite{sheemar2021practical_HYBF}. 
We firstly proceed with the proof by considering the minorized WSR for the digital beamformer $\bmV_{j_b}$, and the proof for the digital beamformer $\bmU_{k_b}$ follows similarly. For the digital beamformer $\bmV_{j_b}$, when the remaining variables are fixed, the following optimization problem can be considered

\begin{equation}
\begin{aligned}
 \underset{\substack{\bmV_{j_b}}}{\max}  \quad & [ w_{j_b} \mbox{ln det}( \bmI +  \bmV_{j_b}^H {\bmW_b}^H \mathbf{\Sigma}_{j_b}^{1} {\bmW_b} \bmV_{j_b}) \\&  -\mbox{Tr}(\bmV_{j_b}^H {\bmW_b}^H  \mathbf{\Sigma}_{j_b}^{2} {\bmW_b} \bmV_{j_b})].
\end{aligned} \label{eq_restate_proof}
\end{equation}
To prove the result stated in Theorem \ref{digital_BF_solution},  \eqref{eq_restate_proof} must be simplified such that the Hadamard's inequality applies as in Proposition 1 \cite{kim2011optimal} or Theorem 1 \cite{hoang2008noncooperative}. The Cholesky
decomposition of the matrix $({\bmW_b}^H \mathbf{\Sigma}_{j_b}^{2} \bmW_b)$ can be written as $\bmL_{j_b} \bmL_{j_b}^H$ where $\bmL_{j_b} $ is the lower triangular Cholesky factor. We can define $\widetilde{\mathbf{V}_{j_b}} = \bmL_{j_b}^H \mathbf{V}_{j_b}$, which allows to write \eqref{eq_restate_proof} as 

\begin{equation} \label{def_1_vj}
\begin{aligned}
 \underset{\substack{\bmV_{j_b}}}{\max}  \quad & [ w_{j_b} \mbox{ln det}( \bmI +  \widetilde{\mathbf{V}_{j_b}}^H \bmL_{j_b}^{-1}   {\bmW_b}^H \mathbf{\Sigma}_{j_b}^{1} {\bmW_b} \bmL_{j_b}^{-H} \widetilde{\mathbf{V}_{j_b}})  \\& -\mbox{Tr}(\widetilde{\mathbf{V}_{j_b}}^H \widetilde{\mathbf{V}_{j_b}} )].
\end{aligned} 
\end{equation}

Let $\bmA_{j_b} \mathbf{\Lambda}_{j_b} \bmA_{j_b}^H $ be the eigen decomposition of $\bmL_{j_b}^{-1}   {\bmW_b}^H \mathbf{\Sigma}_{j_b}^{1} {\bmW_b} \bmL_{j_b}^{-H}$, where  $\bmA_{j_b} $ and $\mathbf{\Lambda}_{j_b}$ are the unitary and diagonal power matrices. Let $\bmD_{j_b} = \bmA_{j_b}^H  \widetilde{\mathbf{V}_{j_b}}^H \widetilde{\mathbf{V}_{j_b}} \bmA_{j_b} $, and the problem \eqref{def_1_vj} can be restated as

\begin{equation}  
\begin{aligned}
 \underset{\substack{\bmD_{j_b}}}{\max}  \quad & [ w_{j_b} \mbox{ln det}( \bmI +   \bmD_{j_b}  \mathbf{\Lambda}_{j_b})  -\mbox{Tr}(\bmD_{j_b} )].
\end{aligned} \label{eq_restate_proof}
\end{equation}

By using Hadamard’s inequality [Page 233 \cite{cover1999elements}], it can be easily shown that the optimal $\bm{D}_{j_b}$ must be diagonal. Let $\bmV_{j_b} = \bmL_{j_b}^{-H} \bmA_{j_b} \bmD_{j_b}^{\frac{1}{2}}$ and we can write 
 
\begin{equation}
    {\bmW_b}^H \mathbf{\Sigma}_{j_b}^{1} {\bmW_b} = \bmL_{j_b} \bmL_{j_b}^H \bmL_{j_b}^{-H} \bmA_{j_b} \bmD_{j_b}^{\frac{1}{2}} \mathbf{\Lambda}_{j_b} =  {\bmW_b}^H  \mathbf{\Sigma}_{j_b}^{2} {\bmW_b} \bmV_{j_b} \mathbf{\Lambda}_{j_b}
\end{equation}
which shows that the WSR maximizing beamformer is made of the GDEs solution. A similar proof can be followed also to prove that the optimal digital beamformer $\bmU_{k_b}$ can be optimized as the GDEs solution.
\vspace{-3mm}
\section{Proof of Theorem \ref{analog_BF_solution}} \label{AP.2}
To optimize the analog beamformer $\bmW_b$ given the remaining variables to be fixed, we consider the following optimization problem
\begin{equation}  \label{analog_t1}
\begin{aligned}
  \underset{\substack{\bmW_b}}{\max}   \sum_{j_b \in \mathcal{D}_b}  & [w_{j_b} \mbox{ln det}(\bmI + \bmV_{j_b}^H {\bmW_b}^H \mathbf{\Sigma}_{j_b}^{1} \bmW_b \bmV_{j_b} )  \\& -\mbox{Tr}(\bmV_{j_b}^H {\bmW_b}^H \mathbf{\Sigma}_{j_b}^{2}  \bmW_b \bmV_{j_b} )].
\end{aligned}
\end{equation}

To prove the result for the analog beamformer, note that a similar reasoning as for the digital beamformer in Appendix A does not apply directly. This is due to the fact that the KKT condition for the analog beamformer
\eqref{kkt_analog_beamformer} have the form
$\bm{A}_1 \bmW_{b} \bm{A}_2 = \bm{B}_1 \bmW_b \bm{B}_2$ and thus not resolvable for $\bmW_b$. To solve it for the analog beamformer $\bmW_b$, we can apply the identity $\mbox{vec}(\bm{A} \bm{X} \bm{B}) = \bm{B}^T \otimes \bm{A}\mbox{vec}(\bm{X})$ \cite{magnus2019matrix}, which allows to rewrite \eqref{analog_t1} as a function of the vectorized analog beamformers $\bmW_b$, leading to a resolvable KKT condition for $\mbox{vec}(\bmW_b)$. A similar proof from Appendix A, as done in Theorem 3 \cite{sheemar2021practical_HYBF}, can be applied to the vectorized analog beamformer $\bmW_b$, and show that the optimal vectorized analog beamformer is a GDE of the sum of the matrices in \eqref{analog_BF}.

\ifCLASSOPTIONcaptionsoff
  \newpage
\fi

{\footnotesize
\bibliographystyle{IEEEtran}
\def\baselinestretch{0.9}
\bibliography{main}}

%\end{IEEEbiography}



\end{document}